\documentclass[british]{amsart}
\usepackage[T1]{fontenc}
\usepackage[utf8]{inputenc}
\usepackage{geometry}
\usepackage{amsthm}
\usepackage{amstext}
\usepackage{graphicx}
\usepackage{amssymb}
\usepackage{esint}

\makeatletter
\numberwithin{equation}{section}
\numberwithin{figure}{section}

\newtheorem{lemma}{Lemma}

\title{Nodal domains of a non-separable problem - the right angled isosceles triangle}

\author{Amit Aronovitch$^{1}$}
\address{$^{\text{1}}$Department of Physics of Complex Systems,
The Weizmann Institute of Science, 76100 Rehovot, Israel,
Email: aronovitch@gmail.com}
\author{Ram Band$^{1,\,2}$}
\address{$^{\text{2}}$ School of Mathematics, University of Bristol, Bristol BS8 1TW, UK, 
Email: rami.band@bristol.ac.uk}
\author{David Fajman$^{3}$}
\address{$^{\text{3}}$Max Planck Institute for Gravitational Physics
  (Albert Einstein Institute), Am M\"uhlenberg 1,D-14476 Golm,
  Germany, 
Email: David.Fajman@aei.mpg.de}
\author{Sven Gnutzmann$^{4}$}
\address{$^{\text{4}}$ School of Mathematical Sciences, University of
  Nottingham, University Park, Nottingham NG7 2RD, United Kingdom, 
Email: Sven.Gnutzmann@nottingham.ac.uk}

\newcommand{\sign}{\mathop{\mathrm{Sign}}\nolimits}

\usepackage{babel}

\makeatother

\usepackage{babel}

\begin{document}

\begin{abstract}
We study the nodal set of eigenfunctions of the Laplace operator on
the right angled isosceles triangle. A local analysis of the nodal
pattern provides an algorithm for computing the number $\nu_{n}$
of nodal domains for any eigenfunction. In addition, an exact recursive
formula for the number of nodal domains is found to reproduce all
existing data. Eventually we use the recursion formula to analyse
a large sequence of nodal counts statistically. Our analysis shows
that the distribution of nodal counts for this triangular shape has
a much richer structure than the known cases of regular separable
shapes or completely irregular shapes. Furthermore we demonstrate
that the nodal count sequence contains information about the periodic
orbits of the corresponding classical ray dynamics.


\end{abstract}
\maketitle

\section{Introduction}

More than 200 years ago Ernst Chladni pioneered the study of standing
waves with his experiments on sound figures for the vibration modes
of plates \cite{chladni}. The sound figures revealed that one may
characterize the modes by looking at the nodal set -- the lines on
the plate which do not take part in the vibration and which are visualized
in a sound figure. For each mode he drew the nodal pattern and counted
the number of nodal lines and nodal domains. With his work he did
not only lay the foundations of modern acoustics but also started
a thread in theoretical and mathematical physics which lead to such
classic results as Sturm's oscillation theorem \cite{sturm} and which
continues to this day.\\
 The mathematical framework starts with the Laplacian $\Delta$
on a compact Riemannian manifold $\mathcal{M}$ -- for the purpose
of this paper it will be sufficient to consider dimension two. If
the manifold has a boundary then Dirichlet boundary conditions will
be assumed. One studies the eigenvalue problem \begin{equation}
-\Delta\varphi=\lambda\varphi\qquad\varphi|_{\partial\mathcal{M}}=0\ .\label{general_laplacian}\end{equation}
 The solutions define the discrete spectrum of (non-negative) eigenvalues
$\{\lambda_{N}\}_{N=1}^{\infty}$ which we assume to be ordered $0\le\lambda_{1}\le\lambda_{2}\le\dots$.
The corresponding eigenfunctions will be denoted $\varphi_{N}$. A
nodal domain of the eigenfunction $\varphi_{N}$ is a connected region
in $\mathcal{M}$ where the sign of $\varphi_{N}$ does not change.
We define the nodal count $\nu_{N}$ as the number of nodal domains
in $\varphi_{N}$. The nodal counts $\{\nu_{N}\}$ form a sequence
of integer numbers which characterizes the vibration modes $\varphi_{N}$
on the shape $\mathcal{M}$. In case of degeneracies in the spectrum
the nodal count is not uniquely defined. This may be overcome in various
ways, e.g. by fixing a basis (and an order in each degeneracy class).
Some results on nodal counts are valid for any choice for the basis
of eigenfunctions -- a famous example is the classic theorem by Courant
\cite{courant} which states $\nu_{N}\le N$.

More recently it has been proposed \cite{BLGNSM02} that one may use
the nodal count sequence to distinguish between \textit{i.} regular
shapes where the Laplacian is separable and the corresponding ray
(billiard) dynamics is integrable, and \textit{ii.} irregular shapes
where the ray dynamics is completely chaotic (see also \cite{nodal-keat,nodal-exp1,nodal-exp2,aiba,SMSA05}).\\
 In the regular separable case the nodal set has a checker board
pattern with crossing nodal lines. The nodal count can easily be found
using Sturm's oscillation theorem in both variables. In this case
many properties of the nodal count sequence can be developed analytically
-- e.g. the statistical distribution of the scaled nodal count $\xi_{N}=\nu_{N}/N$
can be described by an explicit limiting function $P(\xi)$. This
function has some generic universal features: $P(\xi)$ is an increasing
function with support $0\le\xi\le\xi_{\mathrm{crit}}<1$ where $\xi_{\mathrm{crit}}$
is a system dependent cut-off. Near the cut-off, for $\xi<\xi_{\mathrm{crit}}$
the distribution behaves as $P(\xi)\propto(\xi_{\mathrm{crit}}-\xi)^{-1/2}$.
\\
 In the irregular case no explicit counting functional is known.
In this case the nodal lines generally do not have any intersections
and counting nodal domains relies on numerical algorithms (such as
the Hoshen-Kopelman algorithm \cite{hoshen}) that represent the eigenfunctions
on a grid of finite resolution. The numerical procedure is reliable
if the resolution is high enough to resolve the distance between nodal
lines near avoided intersections \cite{monastra}. For high lying
eigenvalues $\lambda_{N}$ the algorithm is time-consuming due to
the increasing grid-size. The numerical experiments have shown that
a limiting distribution $P(\xi)$ takes the form $P(\xi)=\delta(\xi-\overline{\xi})$
where $\overline{\xi}$ is a universal constant (i.e. it does not
depend on the shape). This and other numerical findings have been
shown to be consistent with a seminal conjecture by Berry \cite{berry-rw}
which states that the statistics of eigenfunctions for an irregular
(chaotic) shape can be modelled by the Gaussian random wave (a superposition
of planar waves with the same wavelengths, random direction and random
phase). Bogomolny and Schmit \cite{BOSC02} realized that the nodal
structure of a two-dimensional random wave may be modelled by a parameter-free
critical percolation model (see also \cite{foltin,bogomolny2}). With
this heuristic model they were able to derive an explicit value for
$\overline{\xi}$ (and other features of the nodal set) with excellent
agreement to all numerical data. One interesting implication of the
critical percolation model is that nodal lines can be described by
SLE which has been checked affirmative in numerical experiments \cite{sle-keat,sle-bog,sle-keat2}.
Meanwhile some features of the nodal count have been proven rigorously
for random waves on a sphere -- these rigorous results imply the $\delta$-type
distribution for $P(\xi)$ (but cannot predict the value $\overline{\xi}$).

Another interesting applications of the nodal count that we will touch
in this paper are inverse questions. Two inverse questions have been
discussed to some detail: \textit{i.} Can one resolve isospectrality
by looking at the additional information contained in the nodal count
\cite{GNSMSO05,BRKLPU07,BRKL08}? \textit{ii.} Can one count the shape
of a drum \cite{GNKASM06,nodaltrace-wittenberg,KASM08,KL09}? In other
words, does the sequence of nodal counts (ordered by increasing eigenvalues)
determine the shape of the manifold $\mathcal{M}$? we refer to the
shape rather than the manifold itself as the nodal counts are invariant
under scaling of $\mathcal{M}$.\\
 Both inverse questions have been answered affirmative for certain
sets of shapes and some cases have been proven rigorously \cite{BRKLPU07,KL09}.
However, most recently the first example of a pair of non-isometric
manifolds with identical nodal sequences was found \cite{BRKL08}.\\
In some cases it could be shown that the geometrical information
is stored in the nodal sequence in a way which is very similar to
the way it is stored in spectral functions. For instance, the nodal
count sequence for regular shapes with a separable Laplacian can be
described by a semiclassical trace formula \cite{GNKASM06,nodaltrace-wittenberg,KASM08}.
This trace formula is very similar to the known trace formulas for
spectral functions -- it is a sum over periodic orbits (closed ray-trajectories)
on the manifold where each term contains geometric information about
the orbit. It has been shown that this trace formula can be used to
count the shape of a surface of revolution \cite{KASM08}.\\
 In the irregular case, the existence of a trace formula is an
open question (unpublished numerical experiments by the authors give
some support to the existence of such a formula).

In the present work we continue the thread of research summarized
above and consider the nodal set of the eigenfunctions of one particular
shape: the right angled isosceles triangle (i.e. the triangle with
angles 45-45-90). While this shape is regular with an integrable ray
dynamics, the Laplacian is not separable.

Our main result is an explicit algorithm for the nodal counts. In
contrast to the numerical algorithm used for irregular shapes our
algorithm is exact and does not rely on a finite resolution representation
of the wave function. Though the algorithm is specific to this shape,
the approach may serve as the first step to generalize explicit formulas
for nodal counts beyond the separable case where very few results
are currently available. Furthermore, we conjecture a recursion formula
that allows very efficient evaluation of nodal counts for high eigenvalues
.

In the remainder of the introduction we will introduce the spectrum
and the basis of eigenfunctions for the right angled isosceles triangle.
In section \ref{sec:nodalpattern} we will discuss the nodal structure
of the eigenfunctions and state the nodal count algorithm and the
recursion formula as our main results. In section \ref{sec:applications}
we apply the nodal count algorithm to compute the distribution $P(\xi)$
of scaled nodal counts, and discuss the consistency of the observed
nodal counts with the existence of a trace formula.

\subsection{Eigenvalues and eigenfunctions of the Laplacian for the right angled
isosceles triangle}

\label{sec:intro_triangle}

Let $\mathcal{D}\subset\mathbb{R}^{2}$ be the right angled isosceles
triangle of area $\pi^{2}/2$. For definiteness we choose the triangle
as \[
\mathcal{D}=\{(x,y)\in[0,\pi]^{2}:\ y\le x\}\ .\]
 The eigenvalue problem is stated by \[
-\Delta\varphi(x,y)=-\left(\partial_{x}^{2}+\partial_{y}^{2}\right)\varphi(x,y)=\lambda\varphi(x,y)\quad\text{with \ensuremath{\left.\varphi(x,y)\right|_{\partial\mathcal{D}}=0}}\]
 The spectrum of eigenvalues is given by \[
\lambda_{m,n}=m^{2}+n^{2}\quad\text{for \ensuremath{m,n\in\mathbb{N}^{*}}and \ensuremath{m>n}}\]
 and the corresponding eigenfunctions \begin{equation}
\varphi_{m,n}(x,y)=\sin(mx)\sin(ny)-\sin(nx)\sin(my)\label{eq:eigenfunction_m_n}\end{equation}
 form a complete orthogonal basis.\\
 We denote the nodal count (the number of nodal domains) for $\varphi_{m,n}(x,y)$
by $\nu_{m,n}$. Let us order the spectrum in increasing order, written
as a sequence $\{\lambda_{N}\}_{N=1}^{\infty}$, such that $\lambda_{N}\le\lambda_{N+1}$.
Here $N\equiv N_{m,n}$ is an integer function of the integers $m$
and $n$ (we will continue to suppress the reference to $m$ and $n$)
and we have used a mild abuse of notation by writing $\lambda_{N}=\lambda_{N_{m,n}}=\lambda_{m,n}$.
The spectrum contains degeneracies of a number-theoretic flavour.
For a $g$-fold degenerate eigenvalue $\lambda_{N}=\lambda_{N+1}=\dots=\lambda_{N+g-1}$
we define $N_{m,n}$ by ordering the degenerate values by increasing
$n$. This ordering is arbitrary and has been chosen for definiteness
-- none of our results here would change with a different choice.\\
 In principle one may also be interested in the nodal patterns
of arbitrary eigenfunctions in a degeneracy class. Indeed many physical
applications may imply that the basis functions $\varphi_{m,n}$ cannot
be regarded as typical as soon as one looks at a degenerate eigenvalue.
However, in this paper we will focus exclusively on the nodal counts
of the basis functions $\varphi_{m,n}$ -- for two reasons: \textit{i.}
Understanding the nodal patterns of arbitrary superpositions of the
basis functions $\varphi_{m,n}$ is a much harder problem which does
not follow naturally from understanding just the basis; \textit{ii.}
this choice of basis is natural for any computations.


\section{The nodal pattern \label{sec:nodalpattern}}

The current section describes the main properties of the nodal pattern
of the eigenfunctions $\varphi_{mn}$. These observations are then
used in subsection \ref{sub:Graphic_algorithm} to infer an exact
algorithm for counting nodal domains in the triangle. Eventually,
we propose a very efficient recursion formula for the nodal counts
of the eigenfunctions $\varphi_{mn}$ in subsection \ref{sub:recursion}.

\subsection{A tiling structure of the nodal lines\label{sub:tiling_cases}}

The eigenvalue problem on the triangle possess some symmetry properties
which are revealed in the nodal pattern of the eigenfunctions, $\varphi_{m,n}$.
Specifically, there are eigenfunctions whose nodal sets show a tiling
structure: 
\begin{enumerate}
\item For $m>n$ with $(m+n)\,\bmod\,2=0$, the eigenfunction $\varphi_{m,n}$
is an antisymmetric function with respect to the line $y=\pi-x$.
This line is therefore part of the nodal set of $\varphi_{m,n}$.
The complementary nodal set decomposes into two isometric patterns,
each from either side of the line. Each of these two patterns is similar
to the nodal set pattern of the eigenfunction $\varphi_{m',n'}$ with
$m'=(m+n)/2$ and $n'=(m-n)/2$ (figure \ref{Figure-tiling_cases}(a)).\\

\item For $m>n$ with $\gcd(m,n)=d>1$ the nodal set of the eigenfunction
$\varphi_{m,n}$ consists of $d^{2}$ identical nodal patterns. Each
of these patterns is contained within a sub triangle and they are
tiled together to form the complete pattern. Each such sub pattern
is similar to the nodal set of the eigenfunction $\varphi_{m',n'}$
for $m'=m/d$ and $n'=n/d$ (figure \ref{Figure-tiling_cases}(b)). 
\end{enumerate}
The observations above follow directly from \eqref{eq:eigenfunction_m_n}.

\begin{figure}[!htp]
\hfill{}%
\begin{minipage}[c]{0.3\columnwidth}%
\includegraphics[scale=0.2]{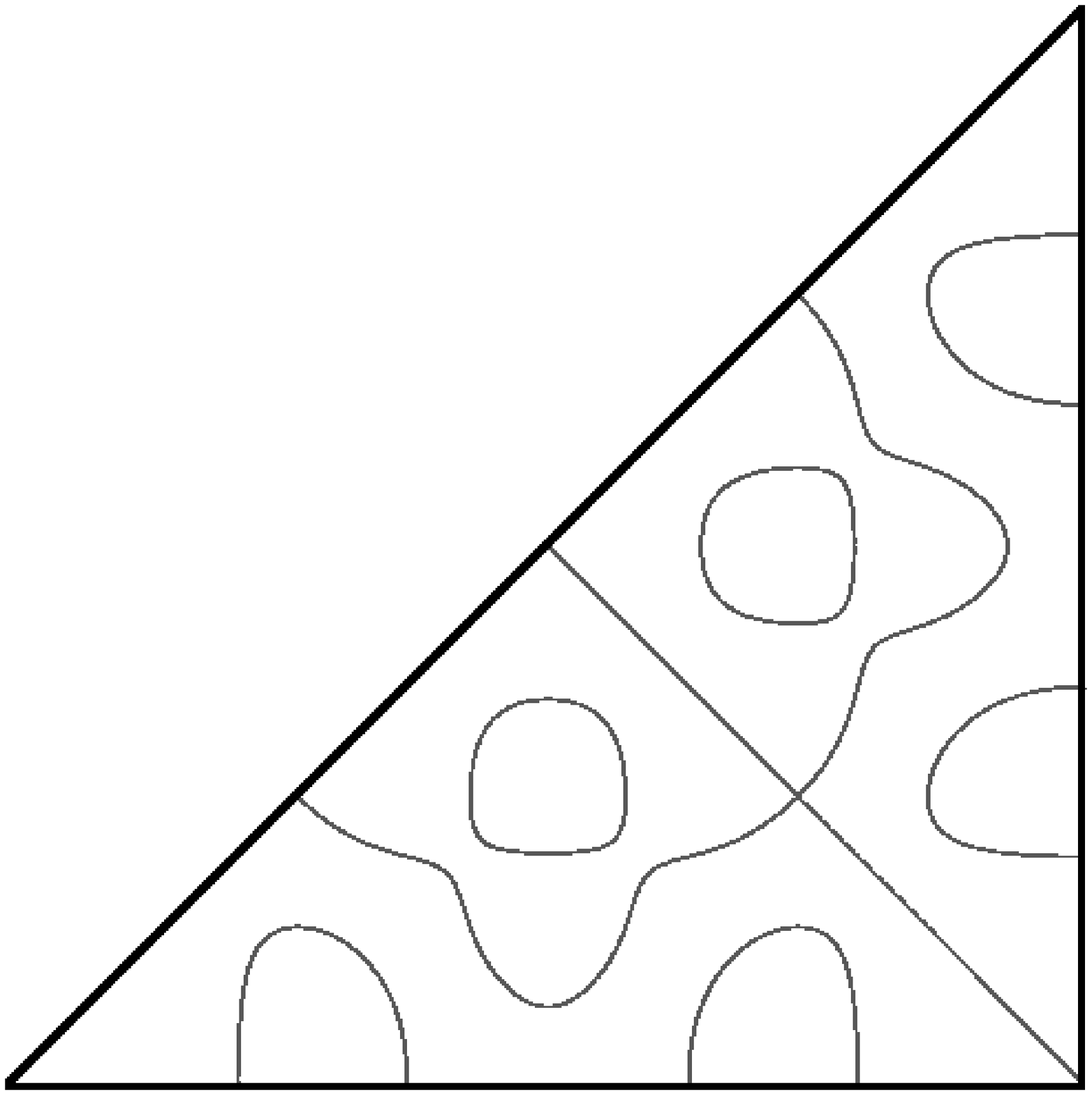}

\medskip{}

\begin{center}
(a) 
\par\end{center}%
\end{minipage}\hfill{}%
\begin{minipage}[c]{0.3\columnwidth}%
\includegraphics[scale=0.2]{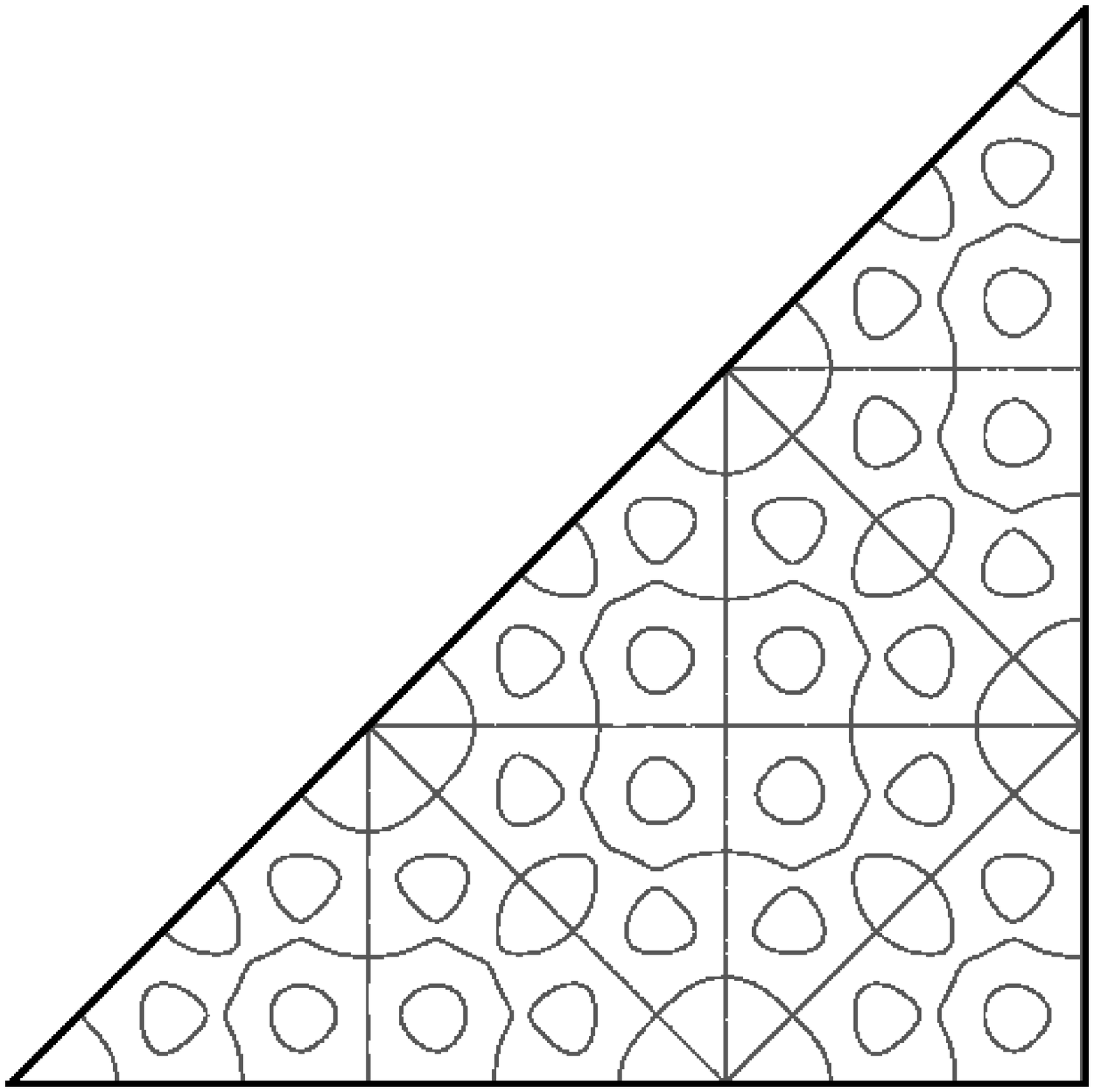}

\medskip{}

\begin{center}
(b) 
\par\end{center}%
\end{minipage}\hfill{}\hfill{}\caption{Two examples for the tiling cases: (a)$\varphi_{9,5}$ and (b)$\varphi_{21,6}$ }

\label{Figure-tiling_cases} %
\end{figure}

\subsection{Characterization of the nodal set}

Let us now characterize the nodal set of the eigenfunction $\varphi_{m,n}$.
We assume that the nodal set of $\varphi_{m,n}$ does not have the
tiling behaviour described in section \ref{sub:tiling_cases}, i.e.
$\gcd(m,n)=(m+n)\bmod\,2=1$. Otherwise, one may reduce the values
of $m,n$, as described above, to a smaller pair $m',n'$, that does
satisfy this condition, and study the nodal set of $\varphi_{m',n'}$
within the reduced triangle. In particular, it is proved in Lemma
\ref{lem:non-tiling-non-crossing} in the appendix that for $m,n$
which satisfy the condition above, the nodal lines of the eigenfunction
$\varphi_{m,n}$ do not cross. This observation is used below to characterize
the nodal set.

We write the eigenfunction $\varphi_{m,n}$ as the difference of the
following two functions \begin{align*}
\varphi_{m,n}^{1}(x,y) & =\sin(mx)\sin(ny),\\
\varphi_{m,n}^{2}(x,y) & =\sin(nx)\sin(my).\end{align*}
 Their nodal sets are correspondingly \begin{align*}
N_{m,n}^{1} & =\left\{ (x,y)\in\mathcal{D}\,\left|\, x\in\frac{\pi}{m}\mathbb{N}\,\vee\, y\in\frac{\pi}{n}\mathbb{N}\right.\right\} ,\\
N_{m,n}^{2} & =\left\{ (x,y)\in\mathcal{D}\,\left|\, x\in\frac{\pi}{n}\mathbb{N}\,\vee\, y\in\frac{\pi}{m}\mathbb{N}\right.\right\} .\end{align*}
 These are regular checkerboard patterns whose nodal domains are open
rectangles and triangles (figure \ref{Fig-checkerboards}(a)).

\begin{figure}
\begin{minipage}[c]{0.3\columnwidth}%
\begin{center}
\includegraphics[scale=0.15]{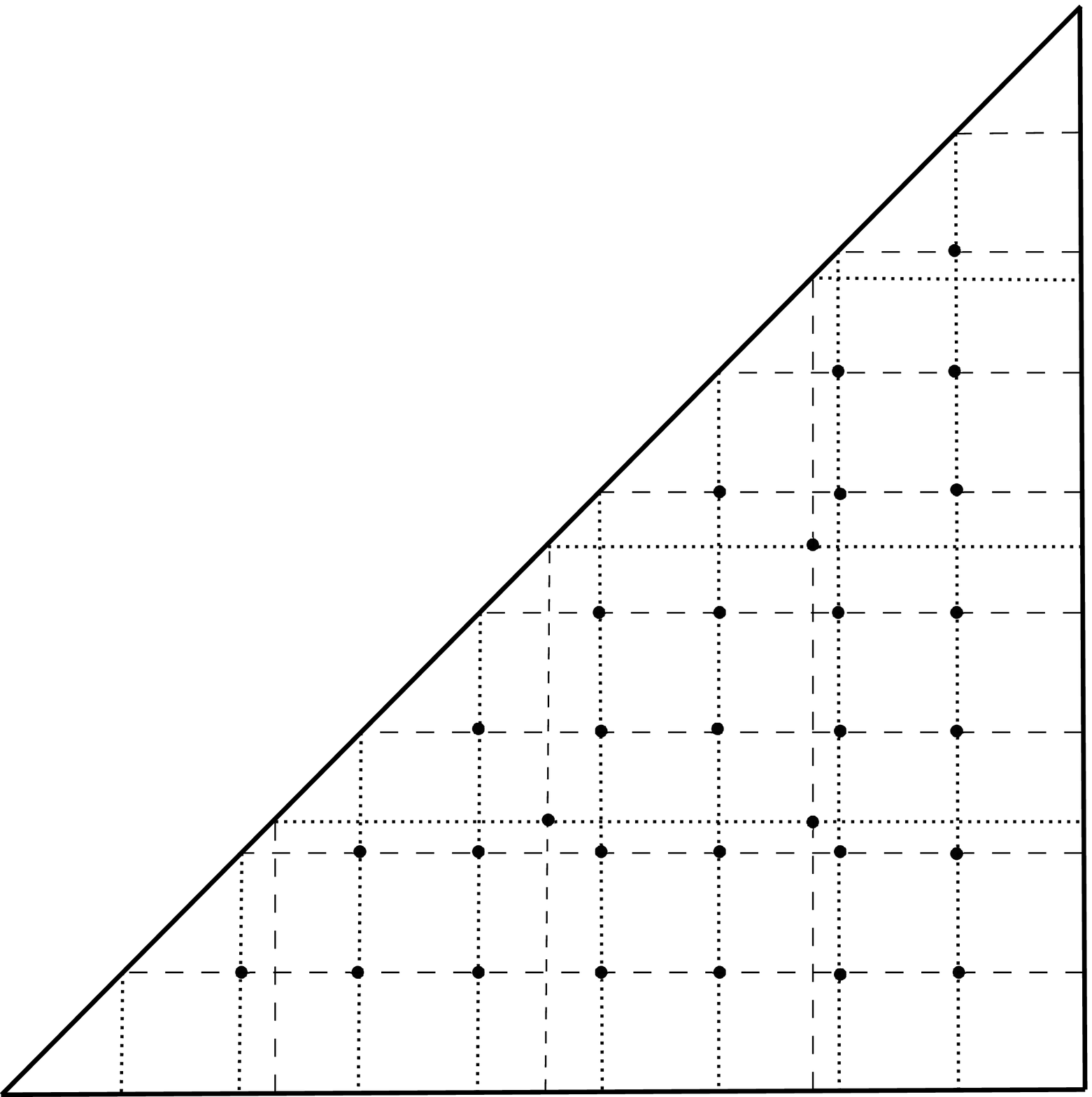} 
\par\end{center}

\medskip{}

\begin{center}
(a) 
\par\end{center}%
\end{minipage}\hspace{1cm}%
\begin{minipage}[c]{0.3\columnwidth}%
\begin{center}
\includegraphics[scale=0.15]{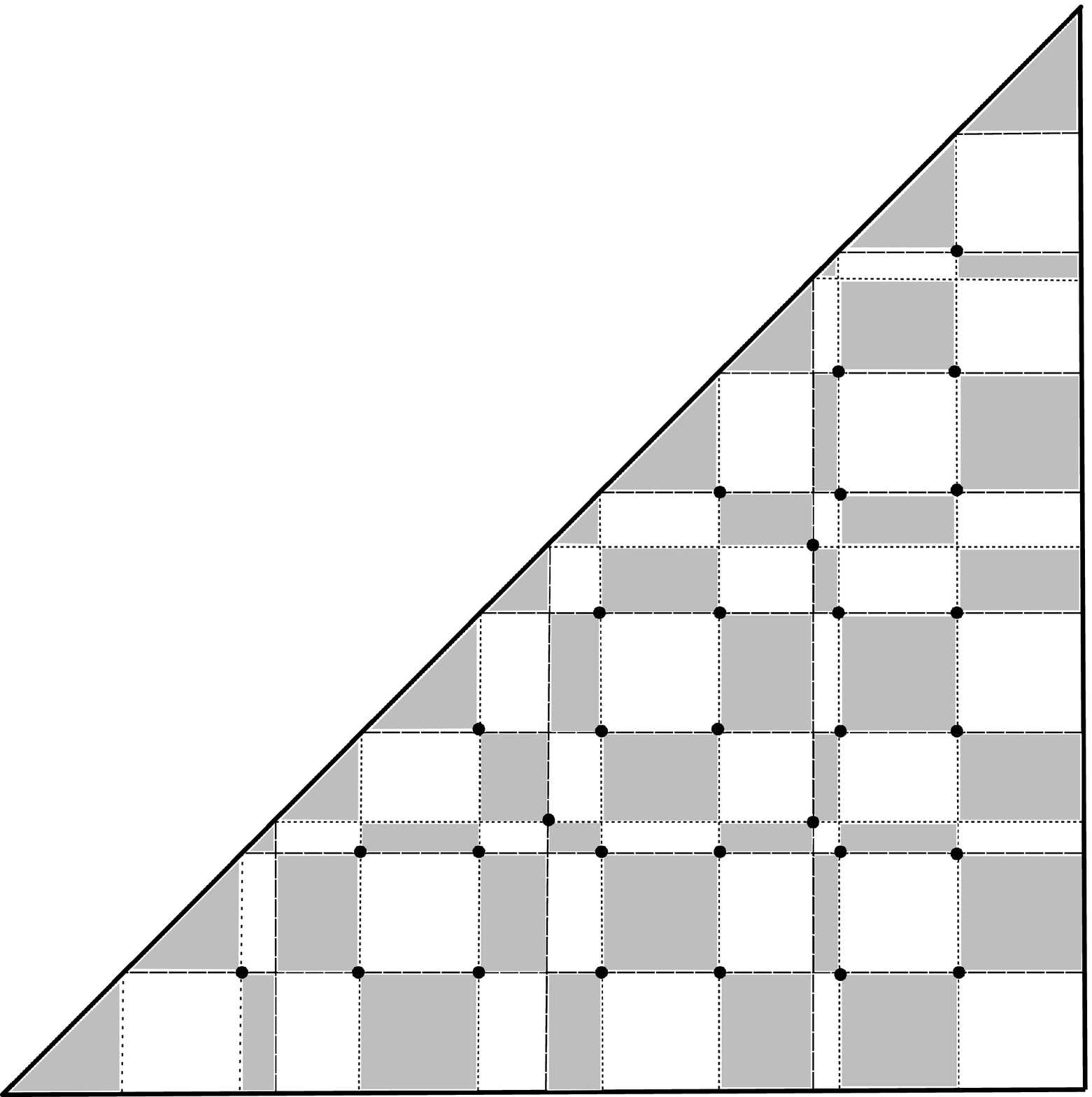} 
\par\end{center}

\medskip{}

\begin{center}
(b) 
\par\end{center}%
\end{minipage}\caption{(a) The nodal sets $N_{9,4}^{1}$ (dotted lines) and $N_{9,4}^{2}$
(dashed lines). (b) The subdomains where $\varphi_{9,4}^{1}$ and
$\varphi_{9,4}^{2}$ have the same sign.}

\label{Fig-checkerboards} %
\end{figure}

The intersection $N_{m,n}^{1}\cap N_{m,n}^{2}$ is the set of points
\[
V_{m,n}=\left\{ \frac{\pi}{m}\left(i,j\right)\,|\,0<j<i<m\right\} \cup\left\{ \frac{\pi}{n}\left(i,j\right)\,|\,0<j<i<n\right\} \]
 (marked with bold points in figure \ref{Fig-checkerboards}). The
eigenfunction $\varphi_{mn}$ vanishes at these points. Hence, nodal
lines pass through them. In the following we analyse the run of the
nodal lines of $\varphi_{mn}$ between the points of the set $V_{m,n}$.
The union $N_{m,n}^{1}\cup N_{m,n}^{2}$ divides $\mathcal{D}$ into
cells shaped as rectangles and triangles of various sizes. These cells
are the connected components of $\mathcal{D}\backslash\left(N_{m,n}^{1}\cup N_{m,n}^{2}\right)$.
The nodal set of $\varphi_{mn}$ is contained within the cells in
which $\varphi_{m,n}^{1}$ and $\varphi_{m,n}^{2}$ have the same
sign. These cells are interlacing in the checkerboard pattern formed
by $N_{m,n}^{1}\cup N_{m,n}^{2}$. We call them the shaded cells and
they appear so in figure \ref{Fig-checkerboards}(b).\\
 The connection between the points in $V_{m,n}$ by nodal lines
can be easily determined by going over the shaded cells and distinguishing
between the following cases: 
\begin{enumerate}
\item A rectangular cell adjacent to two points of $V_{m,n}$. A non self-intersecting
nodal line connects these two points. This is proved in Lemma \ref{lemapp1}.
An example is shown in figure \ref{Fig-connectivity_cases}(a). 
\item A rectangular cell adjacent to four points of $V_{m,n}$. Two nodal
lines connect the two pairs of vertices in either a horizontal or
a vertical non-crossing pattern. One can determine whether the pattern
is horizontal or vertical by comparing the sign of $\varphi_{m,n}$
at the middle point of the rectangle with the sign of $\varphi_{m,n}$
at one of the neighbouring cells. This is proved in Lemma \ref{lemapp1}.
This lemma also proves that a non-tiling eigenfunction, $\varphi_{m,n}$,
cannot vanish at the middle point of the rectangular cell. An example
is shown in figure \ref{Fig-connectivity_cases}(b). 
\item A cell adjacent to a single point of $V_{m,n}$. This happen only
for a cell which is adjacent to the boundary of $\mathcal{D}$. The
$V_{m,n}$ point is then connected to the boundary of $\mathcal{D}$
by a simple non intersecting nodal line. This is proved in Lemma \ref{lemapp2}.
An example is shown in figure \ref{Fig-connectivity_cases}(c). 
\item A triangular cell which do not contain any point of $V_{m,n}$. In
this case there is no nodal line which passes through this triangle.
This is proved in Lemma \ref{lemapp2}.

\end{enumerate}
\begin{figure}
\begin{minipage}[c]{0.3\columnwidth}%
\begin{center}
\includegraphics[scale=0.15]{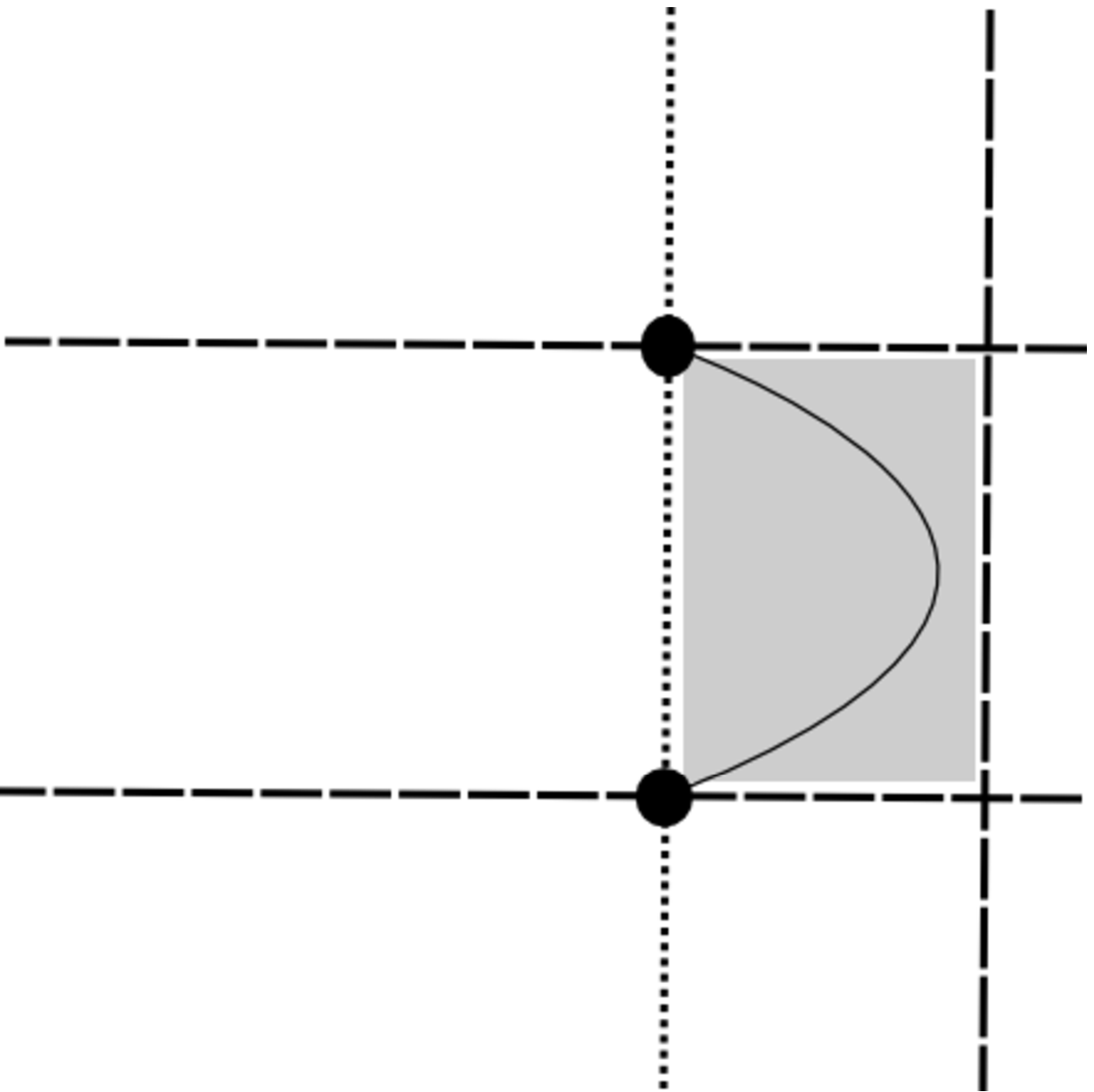} 
\par\end{center}

\medskip{}

\begin{center}
(a) 
\par\end{center}%
\end{minipage}\hfill{}%
\begin{minipage}[c]{0.3\columnwidth}%
\begin{center}
\includegraphics[scale=0.15]{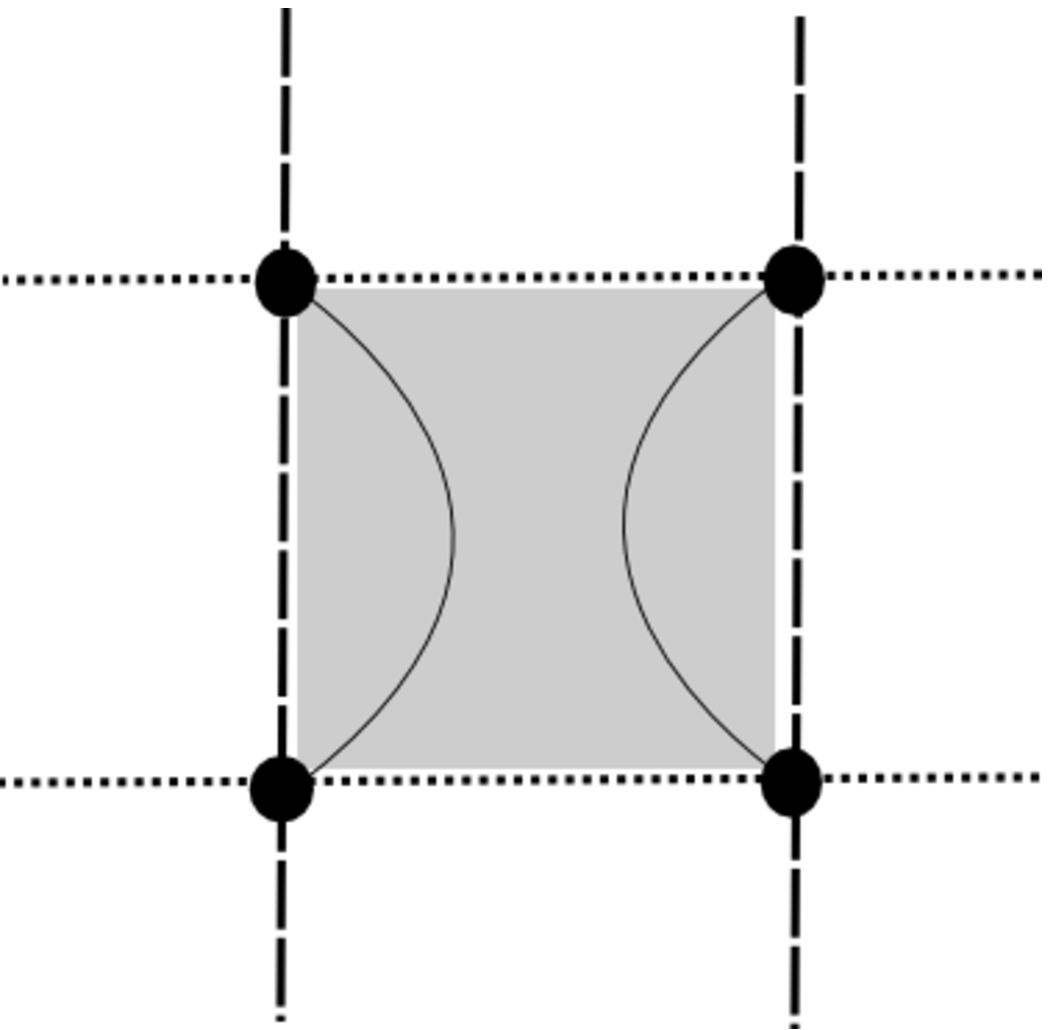} 
\par\end{center}

\medskip{}

\begin{center}
(b) 
\par\end{center}%
\end{minipage}\hfill{}%
\begin{minipage}[c]{0.3\columnwidth}%
\begin{center}
\includegraphics[scale=0.15]{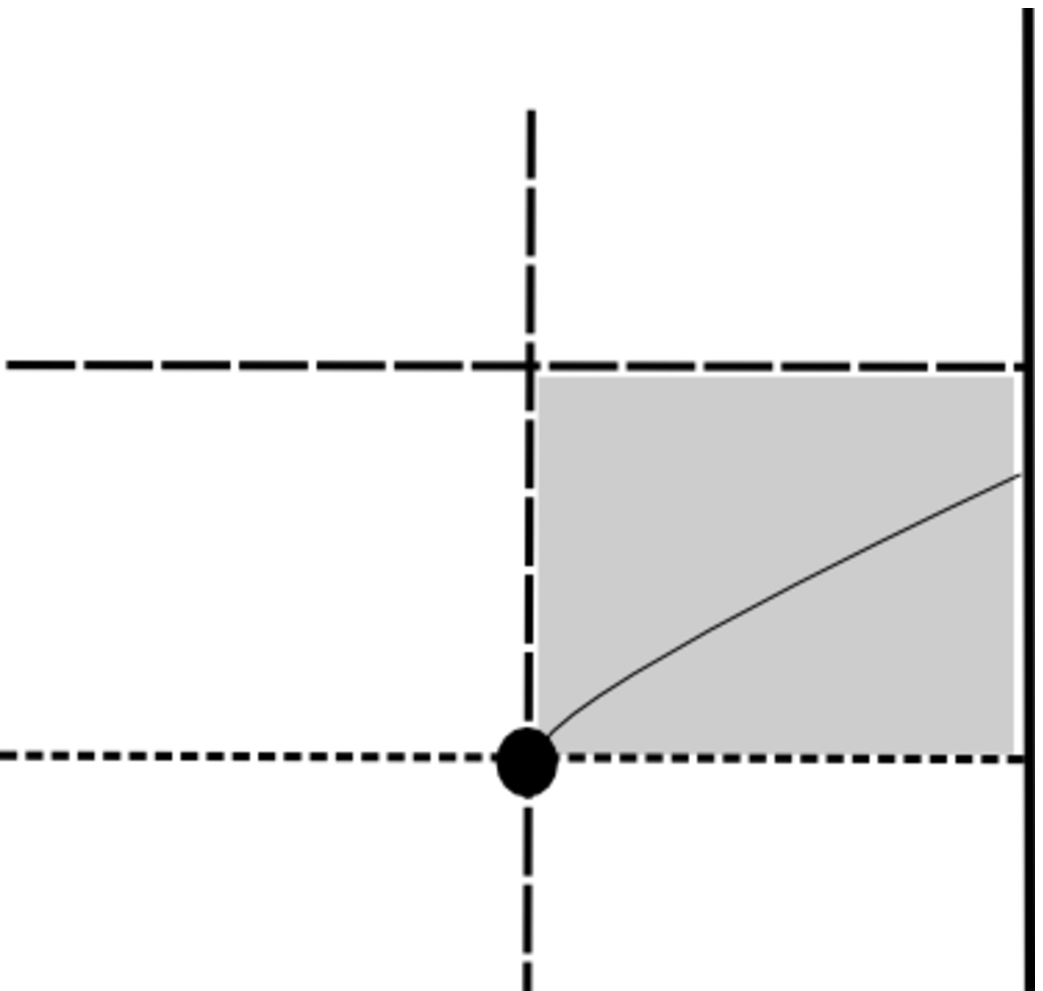} 
\par\end{center}

\medskip{}

\begin{center}
(c) 
\par\end{center}%
\end{minipage}\caption{Different cases of connecting $V_{m,n}$ within the shaded sub-domains
(a) A rectangle with two points from $V_{m,n}$ (b) A rectangle with
four points from $V_{m,n}$ (c) A rectangle with a single point from
$V_{m,n}$}

\label{Fig-connectivity_cases} %
\end{figure}

\subsection{An algorithm for counting the nodal domains\label{sub:Graphic_algorithm}}

We now describe an algorithm for counting $\nu_{m,n}$, the number
of nodal domains of $\varphi_{m,n}$, based on the observations of
the previous section. If the values of $m,n$ correspond to an eigenfunction
with a tiling behaviour we replace them by their reduced values: 
\begin{enumerate}
\item For $m>n$ with $\gcd(m,n)=d>1$, set the new values of $m,n$ to
be $m'=m/d$ and $n'=n/d$. Set the number of tiles to be $d^{2}$.\\

\item For $m>n$ with $(m+n)\,\bmod\,2=0$, set the new values of $m,n$
to be $m'=(m+n)/2$ and $n'=(m-n)/2$. Set the number of tiles to
be 2. 
\end{enumerate}
The number of nodal domains $\nu_{m,n}$ for the original values of
$m,n$ equals to the number of tiles times the number of nodal domains
of the reduced values. We now proceed, assuming the values of $m,n$
were reduced. We create a graph, $G_{m,n}$, whose vertices are $V_{m,n}$
with an additional anchor vertex, $v_{0}$, which stands for the boundary
of the triangle, $\partial\mathcal{D}$. The edges of the graph would
stand for the nodal lines which connect the vertices of $V_{m,n}$.
We go over all shaded cells as described above and for each of them
add either zero, one or two edges to the graph connecting the relevant
vertices. The number of vertices in a cell determines their connectivity,
as described in the previous section%
\footnote{In addition, sampling of $\varphi_{m,n}$ might be required in the
case of a cell adjacent to four vertices.%
}. The cells which contain a nodal line connected to the boundary $\partial\mathcal{D},$
would contribute a single edge to the graph connecting the relevant
vertex of $V_{m,n}$ to the vertex $v_{0}$. Figure \ref{Fig:Pattern_and_Graph}
demonstrate the graph $G_{m,n}$ which corresponds to a certain nodal
set pattern.

\begin{figure}
\begin{minipage}[c]{0.3\columnwidth}%
\begin{center}
\includegraphics[scale=0.2]{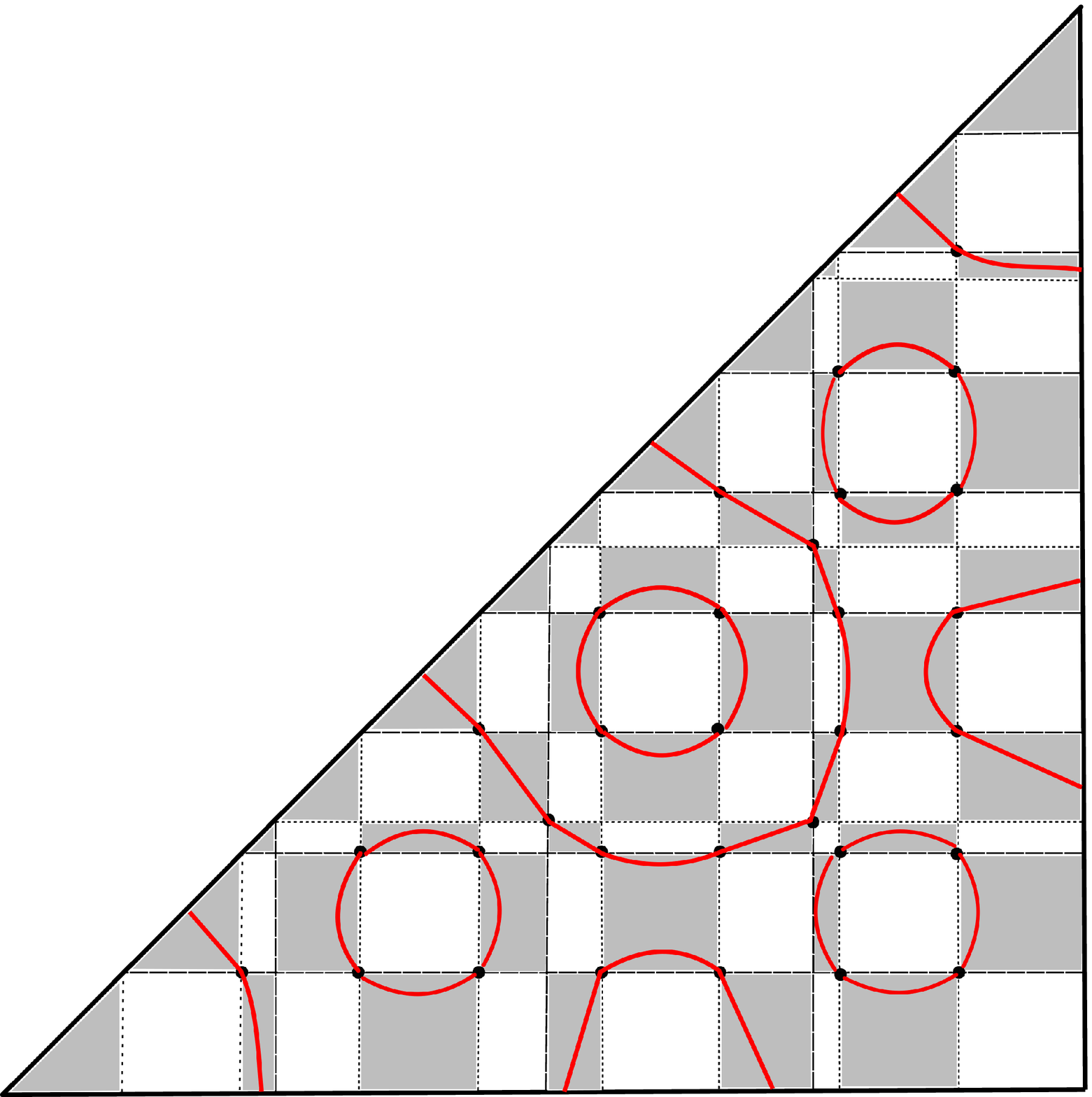} 
\par\end{center}

\medskip{}

\begin{center}
(a) 
\par\end{center}%
\end{minipage}\hspace{1cm}%
\begin{minipage}[c]{0.3\columnwidth}%
\begin{picture}(100,120) \put(0,10){\includegraphics[scale=0.2]{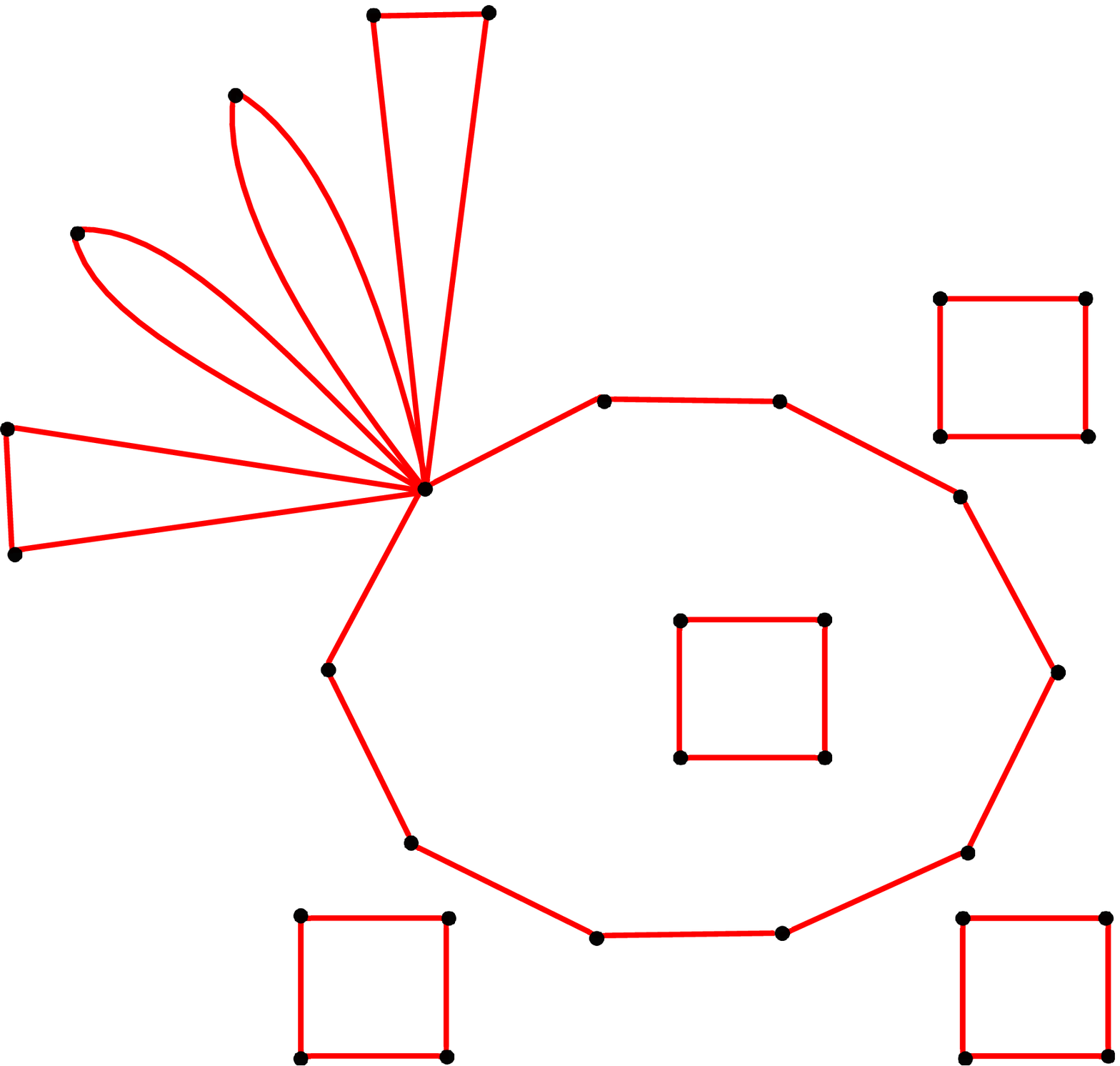}}
\put(40,55){$v_{0}$} \end{picture}

\medskip{}

\begin{center}
(b) 
\par\end{center}%
\end{minipage}

\caption{(a) The nodal set pattern of $\varphi_{9,4}$. (b) The graph $G_{9,4}$
which is produced by the counting algorithm.}

\label{Fig:Pattern_and_Graph} %
\end{figure}

Once the graph $G_{m,n}$ is constructed, the number of nodal domains,
$\nu_{m,n}$, is given by the number of interior faces of the graph
plus one. According to Euler's formula for planar graphs, the number
of faces of $G_{m,n}$ equals $E\left(G_{m,n}\right)-\left|V_{m,n}\right|+c\left(G_{m,n}\right)$,
where $E\left(G_{m,n}\right)$ is the number of edges of $G_{m,n}$
and $c\left(G_{m,n}\right)$ is the number of its connected components.
We therefore get \[
\nu_{m,n}=1+E\left(G_{m,n}\right)-\left|V_{m,n}\right|+c\left(G_{m,n}\right),\]
 which completes the algorithm once $c\left(G_{m,n}\right)$ is calculated
($E\left(G_{m,n}\right)$ and $\left|V_{m,n}\right|$ are known at
this stage).

\subsection{Boundary intersections and nodal loops\label{sub:BI-and-nodal-islands}}

Above we have discussed the nodal count $\nu_{m,n}$. We now introduce
two further quantities which reflect the nodal set structure of the
eigenfunction $\varphi_{m,n}$. The first is the number of intersections
of the nodal set of $\varphi_{m,n}$ with the boundary, $\partial\mathcal{D}$,
which we denote by $\eta_{m,n}$. The second is the number of closed
nodal lines which neither touch the boundary nor intersect themselves
or any other nodal lines. We call those nodal loops, and denote their
number by $I_{m,n}$. In the case where $\varphi_{m,n}$ does not
have a tiling structure, each nodal line is either a loop or a segment
connected to the boundary at two points. Hence, the connection between
the quantities defined above (in the non-tiling case) is given by
the following formula \begin{equation}
\nu_{m,n}=1+\frac{1}{2}\eta_{m,n}+I_{m,n}.\label{eq:nodal_islands_and_intersections}\end{equation}
 As an example, in figure \ref{Fig:Pattern_and_Graph} one can count
$\eta_{9,4}=10$ and $I_{9,4}=4$. The algorithm described in the
preceding section can be used to count $\eta_{m,n}$ and $I_{m,n}$: 
\begin{enumerate}
\item The number of nodal loops, $I_{m,n}$, is given as the number of connected
components of the graph $G_{m,n}$ minus one. 
\item The number of nodal intersections, $\eta_{m,n}$, equals twice the
number of independent cycles of the $G_{m,n}$ component which contains
$v_{0}$. 
\end{enumerate}
It was shown recently (\cite{ARSM09}) that the number of boundary
intersections of the nodal set of $\varphi_{mn}$ in the non-tiling
case is given by \begin{equation}
\eta_{m,n}=m+n-3.\label{f1}\end{equation}

Combining this with \eqref{eq:nodal_islands_and_intersections} indicates
that any formula for the nodal loop count $I_{m,n}$ would yield a
formula for the nodal count $\nu_{m,n}$ and vice versa. 

\subsection{A recursive formula for the nodal loop count\label{sub:recursion}}

In subsection \ref{sub:Graphic_algorithm}, we have described an exact
algorithm that allowed us to compute the nodal loops count. By direct
inspection of tables of evaluated loop counts we have noticed strong
correlations between the counts of different eigenfunctions. An extensive
analysis of such tables allowed us to infer a recursive formula that
we will now describe. Apart from regenerating all data that we looked
at explicitly, we have checked that the empirical formula correctly
predicts all loop counts for the first 100,000 non-tiling eigenfunctions
(this assures agreement of the nodal counts at least up to $N=246062$,
i.e. for all $\varphi_{m,n}$ with $m^{2}+n^{2}\le628325$).\\
 We propose that the loop count $I_{m,n}$ is given by 

\[
I_{m,n}=\tilde{I}\left(n,\,\frac{1}{2}\left(m-n-1\right),\,0\right),\]
where the 3 parameter function $\tilde{I}(n,k,l)$ is defined by the
following recursive formula \begin{equation}
\tilde{I}(n,k,l):=\begin{cases}
0 & n=1\,\text{or}\, k=0\\
\left\lfloor \frac{n}{2k+1}\right\rfloor \left(lk+\left(2l+1\right)k^{2}\right)+\tilde{I}\left(n\,\textrm{mod}\,\left(2k+1\right),\, k,\, l\right) & 2k+1<n\\
\frac{1}{2}\left\lfloor \frac{k}{n}\right\rfloor \left(2l+1\right)\left(n^{2}-n\right)+\tilde{I}\left(n,\, k\,\textrm{mod}\, n,\, l\right) & 2k+1>2n\\
\left(l+\frac{1}{2}\right)\left(2k^{2}+n^{2}-n-2nk+k\right)+\frac{1}{2}k+\tilde{I}\left(2k-n+1,\, n-k-1,\, l+1\right) & n<2k+1<2n.\end{cases}\label{eq:recursive_formula}\end{equation}
 As usual we have assumed that $m,n$ correspond to a non-tiling case
(otherwise, the reduction described above should be made).

\subsubsection*{Remarks}
\begin{enumerate}
\item Note that the description of \eqref{eq:recursive_formula} in terms
of the parameters $\left(n,k\right)=\left(n,\frac{1}{2}\left(m-n-1\right)\right)$
is more compact than a description in terms of the original parameters
$m,n$. 
\item If the initial values of parameters, $n,k$ correspond to a non-tiling
case, i.e. \mbox{$\gcd\left(n+2k+1,n\right)=1$}, then this condition
will hold for all recursive applications of the formula. 
\item One can verify that recursive applications of the formula terminate
at some stage. Namely, that during the recursive applications we arrive
at either $n=1$ or $k=0$. 
\end{enumerate}

\section{Applications to the nodal counting sequence \label{sec:applications}}

\subsection{The nodal count distribution}

Let us now discuss the asymptotic statistics of the number of nodal
domains in terms of the nodal count distribution. 
In section \ref{sec:intro_triangle} we have given a definition of
the nodal count sequence $\{\nu_{N}\}_{N=1}^{\infty}$. Let $\nu_{N}$
be the nodal count 
of the $N$-th eigenfunction. From Courant's nodal domain theorem
\cite{courant} we know that $\nu_{N}\le N$. While the Courant bound
is only realized by a finite number of eigenfunctions \cite{pleijel}
one may still expect that the nodal count will grow $\nu_{N}\sim N$
with the index $N$. It thus makes sense to introduce the scaled nodal
count \begin{equation}
\xi_{N}=\frac{\nu_{N}}{N}\end{equation}
 and ask about the asymptotic behaviour of $\xi_{N}$ as $N\to\infty$.
The latter has been explored by Blum \textit{et al.} \cite{BLGNSM02}
for general two-dimensional billiards in terms of the nodal count
distribution in the interval $\lambda\le\lambda_{N}\le\lambda(1+g)$
for large $\lambda$. The parameter $g>0$ defines the width of the
interval. The limiting distribution is defined as \begin{equation}
P_{\lambda,g}(\xi)=\frac{1}{N(\lambda,g)}\sum_{N:\lambda_{N}\in[\lambda,(1+g)\lambda]}\delta_{\epsilon}(\xi-\xi_{N})\end{equation}
 where $\delta_{\epsilon}(x)=\epsilon\left(\pi(x^{2}+\epsilon^{2})\right)^{-1}$
is a regularized delta-function (the limit $\epsilon\to0$ will always
be implied in the sequel) and $N(\lambda,g)$ is the number of eigenfunctions
in the interval. The integrated distribution will be denoted by \begin{equation}
I_{\lambda,g}(\xi)=\int_{0}^{\xi}P_{\lambda,g,\epsilon}(\xi')d\xi'.\end{equation}
 As mentioned in the introduction an explicit formula for the limiting
distribution \begin{equation}
P(\xi)=\lim_{\lambda\to\infty}P_{\lambda,g}(\xi)\end{equation}
 can be derived for separable Laplacians using semiclassical methods
\cite{BLGNSM02} while for irregular (chaotic) shapes Bogomolny's
percolation model \cite{BOSC02} predicts that the limiting distribution
is concentrated at a universal value $\overline{\xi}$ which is consistent
with all numerical data available. 
The right angled isosceles triangle is neither an irregular shape
(in fact the ray dynamics is integrable) nor are its wave functions
separable. The proposed recursion formula (\ref{eq:recursive_formula})
allows us to find the nodal counts for large sequences of eigenfunctions
very efficiently on a computer. We calculated the nodal counts for
all eigenfunctions with $\sqrt{\lambda_{N}}\le13000$ (about 66 million
eigenfunctions) and extracted the nodal count distributions in various
intervals. In the remainder of this section we will set $g=1$ and
discuss the numerical results. %
\begin{figure}
\centering{}\includegraphics[width=0.75\textwidth]{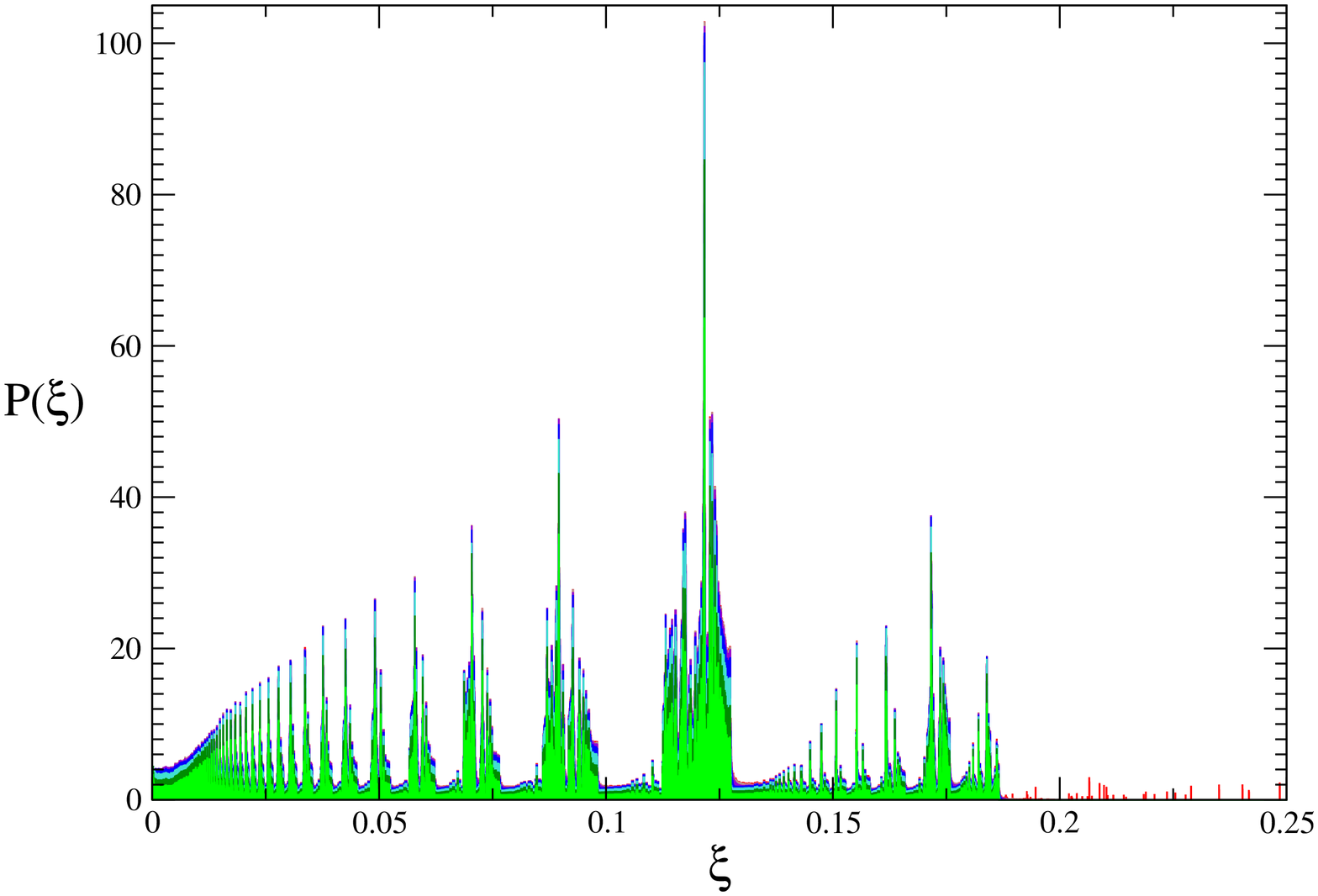}\\
 \includegraphics[width=0.75\textwidth]{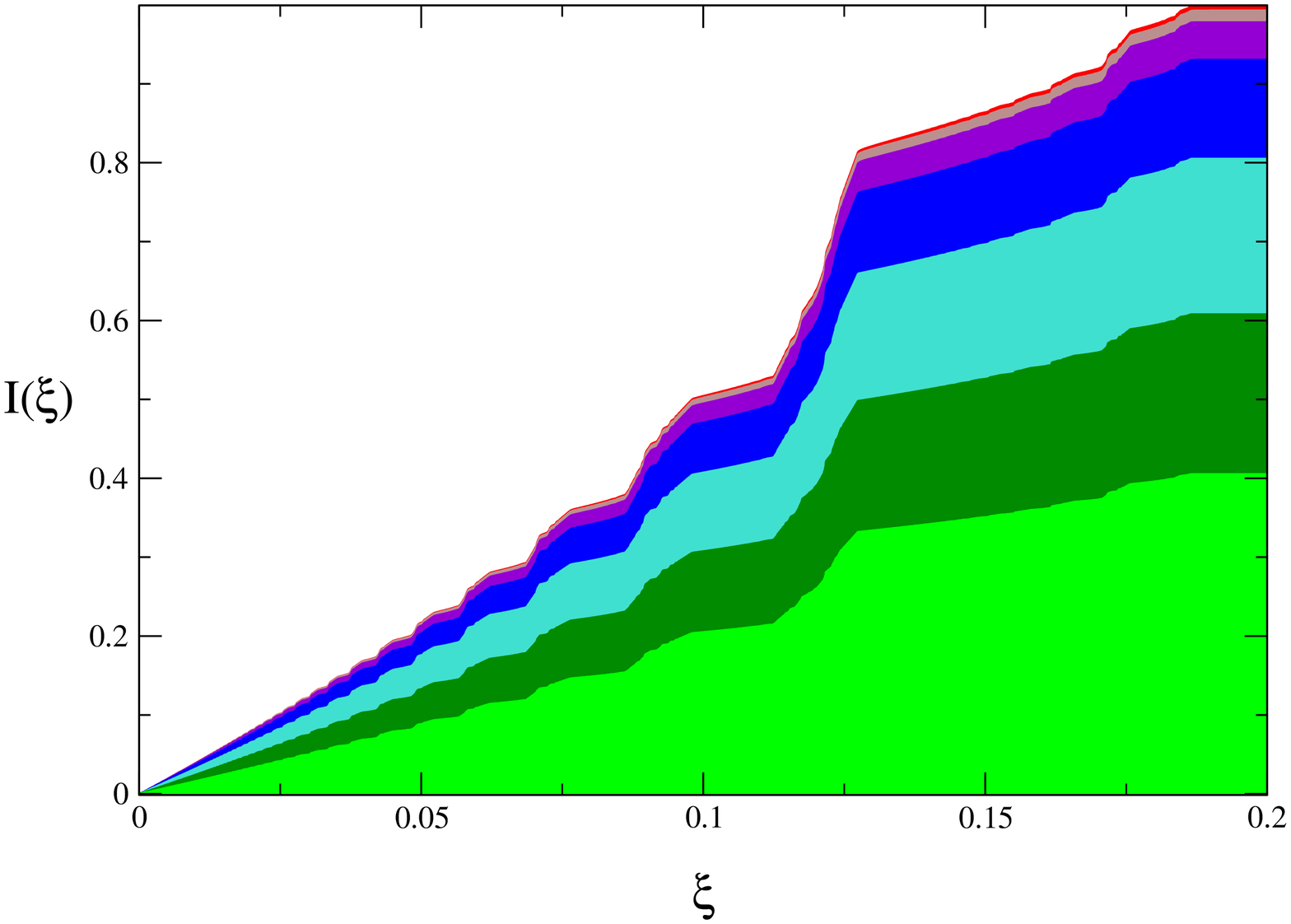}
\caption{ \label{fig:pxi} Upper panel: the nodal count distribution (histogram)
for energies in the interval $9000^{2}\le\lambda_{N}\le2\cdot9000^{2}$.
The colour represent proportion of wave functions with no tiling behaviour
(light green), with exactly 2 tiles (dark green), with 4 to 9 tiles
(turquoise), with 10 to 99 tiles (blue), with 100 to 999 tiles (violet),
with 1000 to 9999 tiles (grey), and with more than 10000 tiles (red).
Lower panel: the corresponding integrated nodal count distribution. }
\end{figure}

Figure \ref{fig:pxi} reveals that the nodal count distribution $P_{\lambda,1}(\xi)$
(with $\lambda=9000^{2}$) for the isosceles triangle contains a lot
of puzzling structure that neither resembles the monotonic behaviour
known from separable billiards nor the single delta-peak known to
describe chaotic billiards. Instead the distribution consists of many
peaks whose strengths and distances form a visible pattern. Each peak
apparently has a further substructure. The same structure appears
if one only includes wave functions without tiling behaviour (or with
a specific number of tiles).\\
 Comparing the nodal count distributions $P_{\lambda,1}(\xi)$
for various values of $\lambda$ gives us some insight into the asymptotic
behaviour of $P_{\lambda,1}(\xi)$. Figure \ref{fig:pxi_comp} shows
how two peaks in the distribution move and change shape as $\lambda$
increases: all peaks move to the left and become sharper. The comparison
reveals that our numerical calculation of $P(\xi)$ has not converged
-- in spite of the extensive number of nodal counts included we cannot
be sure whether a limiting distribution exists. Still it is interesting
to note that, in a certain sense, the asymptotic behaviour contains
some features of chaotic systems. In a chaotic billiard one sees a
single peak which becomes a delta-function as $\lambda\to\infty$.
For the isosceles triangle we see a large number of peaks -- and the
numerics suggests that each one may converge to a delta-function.
Another obvious question suggested by the numerics is whether the
limiting distribution contains fractal features. %
\begin{figure}
\begin{centering}
\includegraphics[width=0.9\textwidth]{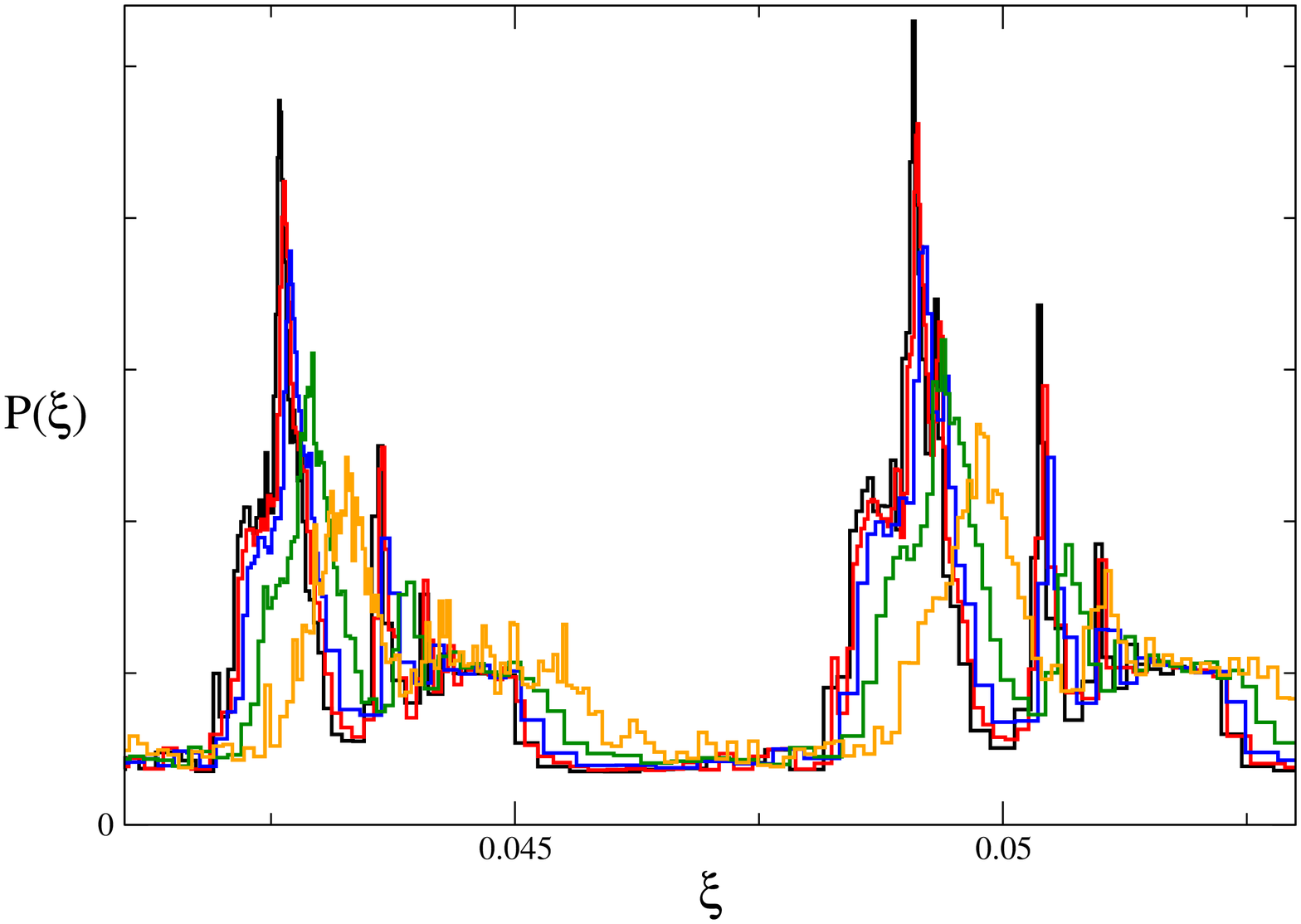} 
\par\end{centering}

\caption{\label{fig:pxi_comp} A detail of the nodal count distribution $P_{\lambda,1}(\xi)$
that shows the limiting behaviour. The five curves are histograms
for $\lambda=1000^{2}$ (orange), $\lambda=2000^{2}$ (green), $\lambda=4000^{2}$
(blue), $\lambda=6000^{2}$ (red), and $\lambda=9000^{2}$ (black). }
\end{figure}

\subsection{The cumulative nodal loop count}

We have already observed in section \ref{sub:BI-and-nodal-islands}
that, at least for the non-tiling case, the nodal count decomposes
to the number of boundary intersections and the nodal loop count \eqref{eq:nodal_islands_and_intersections}.
The number of boundary intersections for the triangle was already
investigated in \cite{ARSM09} and presented as a trace formula. In
this section we thus focus on the nodal loop count. Denoting by $\iota_{n}$
the nodal loop count of the $n$-th eigenfunction, we define two cumulative
continuous counting functions: \begin{align*}
Q(N): & =\sum_{n=1}^{\lfloor N\rfloor}\iota_{n}\\
C(k) & :=\sum_{n=1}^{\infty}\iota_{n}\Theta\left(k-k_{n}\right),\end{align*}
 where $\lfloor N\rfloor$ denotes the largest integer smaller than
$N$, $k_{n}$ is the square root of the $n$-th eigenvalue (multiple
eigenvalues appear more than once in the sequence $\left\{ k_{n}\right\} )$
and $\Theta\left(k\right)$ is the Heaviside theta function. It should
be noted that the functions above can be obtained one from the other
by use of the spectral counting function, $N\left(k\right)=\sum_{n=1}^{\infty}\Theta\left(k-k_{n}\right)$,
or its inversion. Previous works examined similar nodal counting functions
for separable drums \cite{GNKASM06,nodaltrace-wittenberg}. It was
proved that for simple tori and surfaces of revolution the nodal counting
function can be presented as a trace formula. The counting function
was expressed there as a sum of two parts: a smooth (Weyl) term which
reflects the global geometrical parameters of the drum, and an oscillating
term which depends on the lengths of the classical periodic orbits
on the drum. For example, it was shown in \cite{GNKASM06,nodaltrace-wittenberg}
that the smooth part of $\sum_{n=1}^{\lfloor N\rfloor}\nu_{n}$ is
$\mathrm{O}(N^{2})$, and the oscillating term has the form \[
N^{\frac{5}{4}}\sum_{\mathrm{po}}a_{\mathrm{po}}\sin\left(L_{\mathrm{po}}\sqrt{\frac{4\pi}{A}N}+\varphi_{\mathrm{po}}\right),\]
 where the sum is over the periodic orbits, $L_{\mathrm{po}}$ is
the length of the orbit, $a_{\mathrm{po}},\,\varphi_{\mathrm{po}}$
are some coefficients, which depend on the orbit, and $A$ is the
total area of the drum. Results for other separable drums have the
same form.

Having in mind the case of separable drums, we have examined both
$Q\left(N\right)$, and $C\left(k\right)$ numerically and found that
both counting functions have a (numerically) well-defined smooth term
and an oscillatory term. Like in the case of the separable drums,
the smooth term of $C\left(k\right)$ was found to be $\mathrm{O}(k^{4})$
as well. Note that the accumulated boundary intersections count $\sum_{n=1}^{\infty}\eta_{n}\Theta\left(k-k_{n}\right)$
is only ${\rm O}(k^{3})$. Hence, for high energy eigenfunctions,
most of the nodal domains do not touch the boundary. We have extracted
the oscillatory part by numerically interpolating the smooth term
and then subtracting it from $C\left(k\right)$. In order to reveal
whether periodic orbits contribute in a similar way as in the separable
case we evaluated the Fourier transform of the oscillatory term $C_{\mathrm{osc}}\left(k\right)$.
The result is shown in figure \ref{Fig:Fourier_by_k} where the transform
was performed for the interval \[
\left(k_{62439153},\, k_{62831853}\right)\approx\left(\sqrt{9466^{2}+8332^{2}},\,\sqrt{10046^{2}+7688^{2}}\right).\]

\begin{figure}
\includegraphics[clip,scale=0.35]{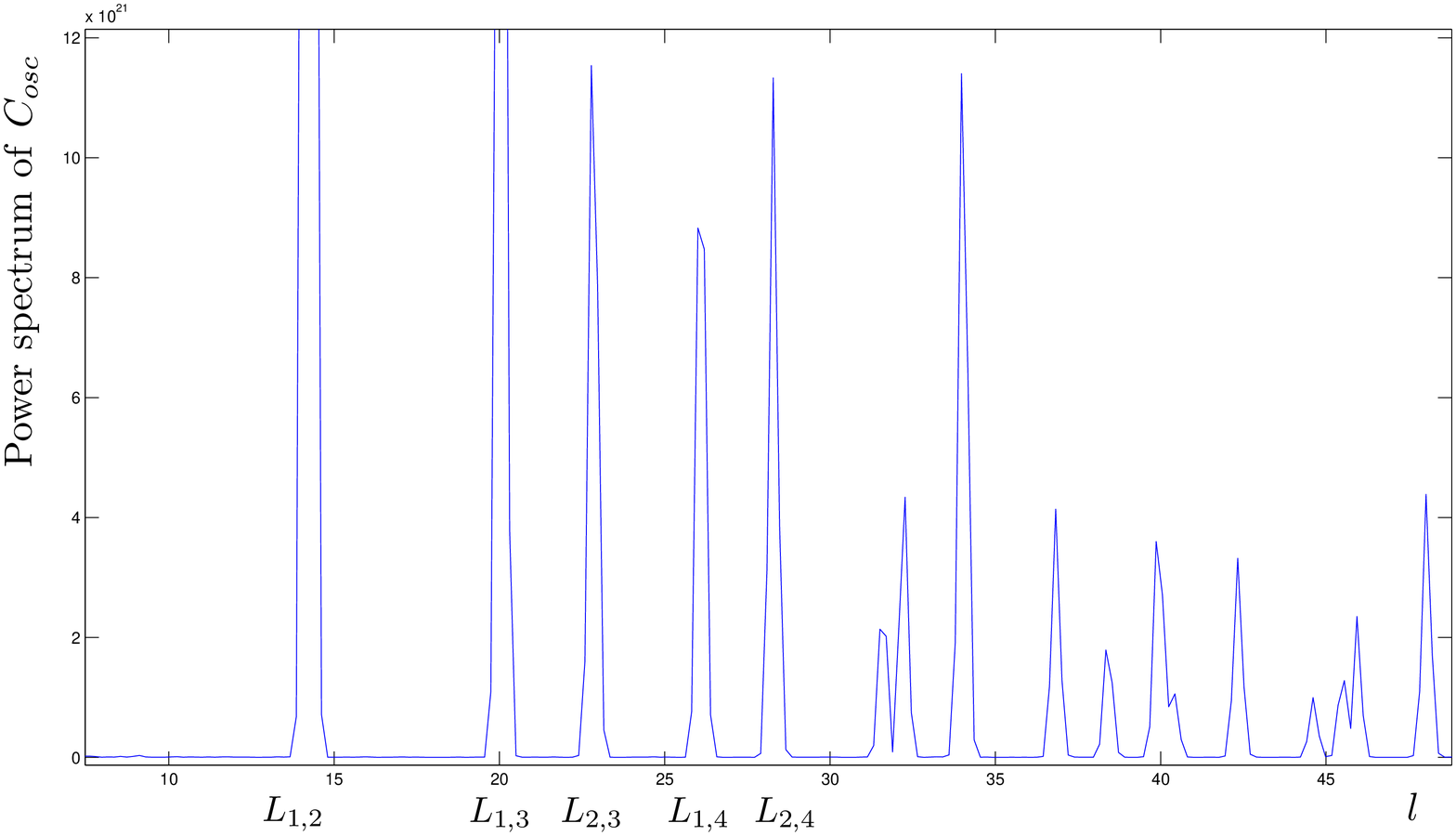}

\caption{The power spectrum of $C_{\mathrm{osc}}\left(k\right)$. The lengths
of some periodic orbits are identified on the $l$ axis.}

\label{Fig:Fourier_by_k} %
\end{figure}

The Fourier transform in figure \ref{Fig:Fourier_by_k} shows clear
peaks at positions which correspond to lengths of periodic orbits
in the triangle. For each value of $\left(p,q\right)\in\mathbb{Z}^{2}\setminus\{(0,0)\}$,
there exists a continuous family of orbits of length $L_{p,q}=2\pi\sqrt{p^{2}+q^{2}}$.
These are orbits that bounce from the bottom cathetus ($y=0$) at
an angle of $\arctan(q/p)$.

The investigation of $Q\left(N\right)$ starts similarly by extracting
its oscillating part. As can be expected from Weyl's formula, the
smooth part is $\mathrm{O}(N^{2})$. However, the Fourier transform
of $Q_{\mathrm{osc}}\left(N\right)$ should be done with respect to
a scaled variable rather than $N$. For this purpose, we use the Weyl
term of the counting function, $N\approx\frac{A}{4\pi}\lambda_{N}$,
where $A=\frac{1}{2}\pi^{2}$ is the area of $\mathcal{D}$ and $\lambda_{N}$
is the $n$-th eigenvalue. The scaled variable used for the Fourier
transform is the square root of the Weyl-estimated eigenvalue, $q\equiv$$\sqrt{\frac{4\pi}{A}N}=\sqrt{\frac{8}{\pi}N}$.
Fourier transforming $Q_{osc}$ with respect to $q$, reveals a linear
combination of delta-like peaks. The positions of these peaks reproduce
the lengths of some of the periodic orbits mentioned above and of
some additional ones: 
\begin{enumerate}
\item Isolated orbits that hit the corner $\left(\pi,0\right)$ at $45^{\circ}$.
The length of such orbits is $L_{n}=\sqrt{2}\pi n$, where $n\in\mathbb{N}$
is the number of repetitions of the basic orbit. 
\item Isolated orbits that go along one of the catheti. Their length is
$\tilde{L}_{n}=2\pi n$, where $n\in\mathbb{N}$ is the number of
repetitions of the basic orbit. 
\end{enumerate}
Figure \ref{Fig:Fourier_by_n} shows the power spectrum of $Q_{\mathrm{osc}}\left(q\right)$,
done when analysing $Q\left(N\right)$ in the interval $N\in\left(38877209,\,39269906\right)$.
\begin{figure}
\includegraphics[clip,scale=0.35]{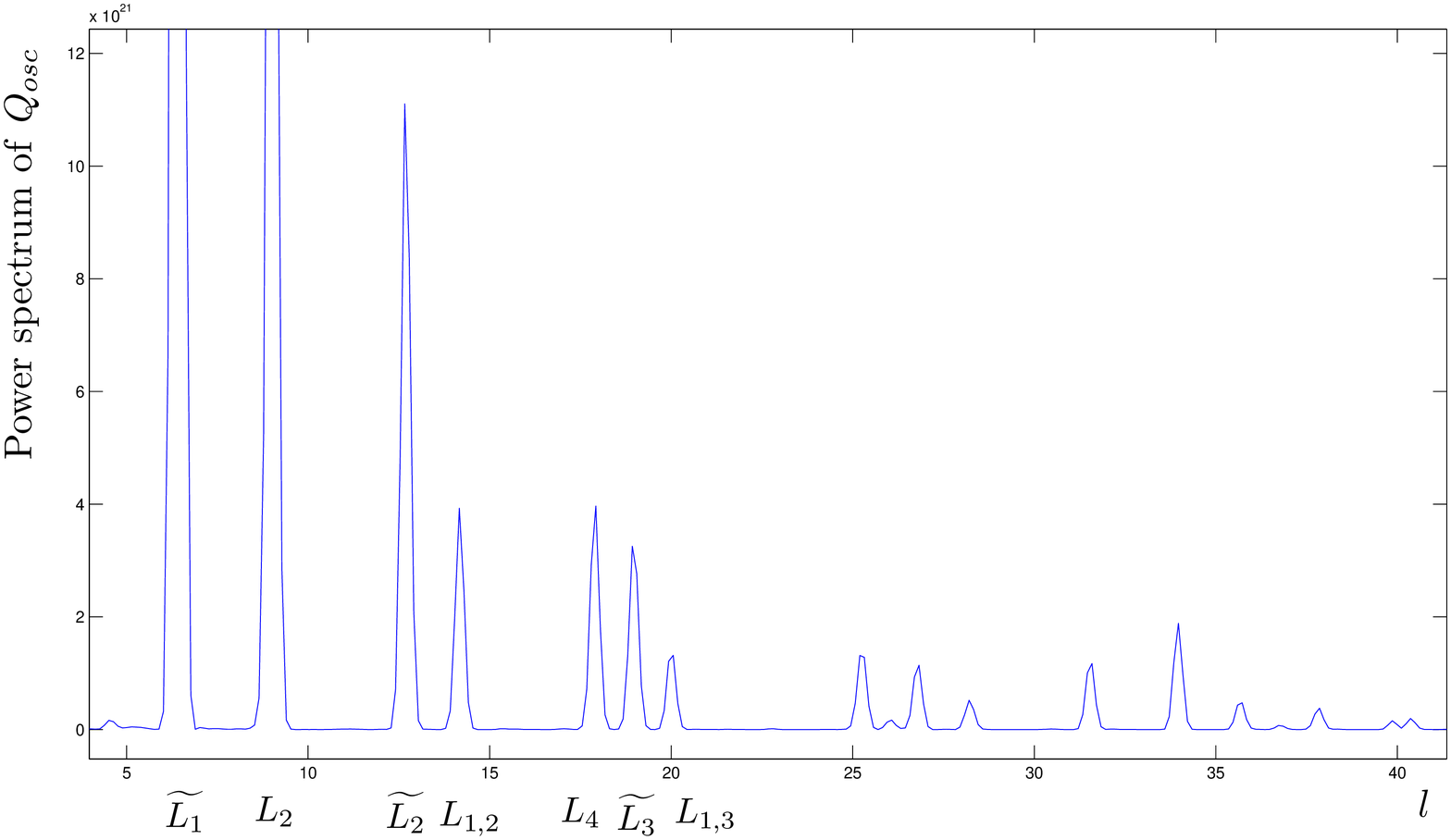}

\caption{The power spectrum of $Q_{\mathrm{osc}}\left(q\right)$. The lengths
of some periodic orbits are identified on the $l$ axis.}

\label{Fig:Fourier_by_n} %
\end{figure}

The above numeric investigation suggests a few observations. The clean
Fourier transforms of both $C_{\mathrm{osc}}$ and $Q_{\mathrm{osc}}$
indicate on the existence of a trace formula for both. The need to
rescale the variable before Fourier transforming $Q_{\mathrm{osc}}$,
suggests that the source of a trace formula for $Q\left(N\right)$
is the trace formula of $C\left(k\right)$ combined with the inversion
of the spectral counting function, $N\left(k\right)$. A similar relation
between the boundary intersections counting functions was revealed
in \cite{ARSM09}. Another observation is that only the continuous
families of periodic orbits appear in the Fourier transform of $C_{\mathrm{osc}}$.
This is fundamentally different from the trace formula of the boundary
intersections (\cite{ARSM09}) and calls for further investigation.
We suggest that the isolated periodic orbits which do appear in the
Fourier transform of $Q_{\mathrm{osc}}$ are caused by the spectral
inversion.

\section{Summary \& Discussion}

This paper investigates the nodal set of the Laplacian eigenfunctions
of the right angled isosceles triangle. The novelty of the work is
the ability to obtain exact results for the nodal count although this
problem is not separable. The algorithm described in section \ref{sub:Graphic_algorithm}
constructs a graph which reflects the topology of the nodal set of
a given eigenfunction. The graph contains complete and exact information
about various properties of the nodal set (such as the number of loops
and the number of nodal domains) which can be calculated straight
forwardly. The standard algorithm used for computing the number of
nodal domains for a known (non-separable) eigenfunction on a drum
is the Hoshen-Kopelman algorithm \cite{hoshen}. It samples the eigenfunction
on a grid of finite resolution. As far as we know all implementations
of the Hoshen-Kopelman algorithm for nodal counting use a fixed grid
and calculate the number of nodal domains as an approximation. In
principle, one may reduce the error by increasing the resolution of
the grid near avoided crossing. However, the application of this algorithm
assumes a priori that there are no nodal intersections. For the special
algorithm we provide here we have proven that it gives the exact result,
even though it samples the eigenfunction more sparsely than the Hoshen-Kopelman.
This also leads to a somewhat faster running time of our algorithm
(for both algorithms the running time is proportional to the energy
$\lambda$ -- however the constant of proportionality is lower for
our algorithm).

Our result may be generalized to other domains where similar algorithms
may apply. Our algorithm is based on the fact that the eigenfunctions
are presented as a linear combination of simple plane waves. It is
therefore tempting to try and generalize it for other drums with similar
property. The equilateral triangle is an immediate candidate (see
\cite{McCartin} and references within).

A further, and quite surprising, result is the recursive formula for
the number of nodal loops. To our knowledge this is the first known
exact formula for the nodal count of a non-separable planar manifold
(for certain eigenfunctions of tori exact formulas have been given
in \cite{BRKL08}). The formula was found by direct inspection of
large tables and has been verified for a large bulk of data computationally.
An obvious challenge is to prove this formula. In particular, the
recursive part of the formula resembles the famous Euclid algorithm
for the greatest common divisor. A further investigation of the mentioned
formula might therefore expose some new number theoretical properties
of the nodal count.

The recursive formula enables us to compute a large amount of data
and to study the statistical properties of the nodal count sequence.
We have studied this sequence using functions which are commonly used
in research of nodal domains: the nodal count distribution and the
cumulative nodal count. The nodal count distribution showed intriguing
structure that resembles neither the behaviour known from separable
billiards nor the one of chaotic billiards. If at all, there is some
similarity to the chaotic case, where the limiting distribution is
a single delta function, whereas in our case it contains a large number
of peaks.\\
 In our analysis of the cumulative nodal count we found numerical
evidence for the existence of a trace formula similar to the one recently
derived for separable drums \cite{GNKASM06,nodaltrace-wittenberg}.
An open question is therefore to prove the existence of a trace formula
in our case, shedding more light on the question `Can one count the
shape of a drum?' 

\section*{Acknowledgements}

We are grateful to Uzy Smilansky for presenting the problem to us
and for the continuous support of the work. We thank Amit Godel for
fruitful discussions. AA and RB thank the University of Nottingham
for hospitality. DF and SG thank the Weizmann Institute for hospitality.
The work was supported by ISF grant 169/09. RB is supported by EPSRC,
grant number EP/H028803/1. DF was supported by the Minerva Foundation.

\newpage{} 

\appendix

\section{Proofs of three lemmas}

\begin{lemma}\label{lem:non-tiling-non-crossing} Let \begin{equation}
\varphi_{m,n}(x,y)=\sin(mx)\sin(ny)-\sin(nx)\sin(my),\label{eq:eigenfunction_appendix}\end{equation}
 be an eigenfunction of the Laplacian on $\mathcal{D}$, where $m,\, n$
obey $\gcd(m,n)=(m+n)\bmod2=1$. Then there are no crossings of the
nodal set of $\varphi_{m,n}$ in the interior of $\mathcal{D}$.

\end{lemma} \begin{proof} The necessary conditions for a crossing
to happen at a point $\left(x,y\right)$ are \begin{align*}
\varphi_{m,n}\left(x,y\right) & =0,\\
\nabla\varphi_{m,n}\left(x,y\right) & =0.\end{align*}

After some algebraic manipulations the equations above give \begin{align}
\frac{\sin\left(nx\right)}{\sin\left(mx\right)} & =\frac{\sin\left(ny\right)}{\sin\left(my\right)},\label{eq:condition1}\\
\frac{\tan\left(nx\right)}{\tan\left(mx\right)} & =\frac{n}{m},\label{eq:condition2}\\
\frac{\tan\left(ny\right)}{\tan\left(my\right)} & =\frac{n}{m}.\label{eq:condition3}\end{align}
 Combining \eqref{eq:condition1} , \eqref{eq:condition2} and \eqref{eq:condition3}
gives \[
\frac{\cos\left(nx\right)}{\cos\left(mx\right)}=\frac{\cos\left(ny\right)}{\cos\left(my\right)}.\]
 Squaring this and using \eqref{eq:condition1} allows to conclude
that one of the following holds \begin{align*}
\sin^{2}\left(my\right) & =\sin^{2}\left(ny\right)\,\,\textrm{or}\\
\sin^{2}\left(my\right) & =\sin^{2}\left(mx\right).\end{align*}
 Assuming $\sin^{2}\left(my\right)=\sin^{2}\left(ny\right)$ immediately
leads to $\frac{n}{m}=\pm1$, which contradicts the assumptions on
the values of $m$ and $n$. Assuming $\sin^{2}\left(my\right)=\sin^{2}\left(mx\right)$
leads to $\sin^{2}\left(ny\right)=\sin^{2}\left(nx\right)$. We are
now required to examine several possibilities for the relations of
the expressions $mx,my,nx,ny$. Such an examination shows that each
possibility will lead to a contradiction with the requirements $x,y\in\left(0,\pi\right)$
and the conditions $\gcd(m,n)=(m+n)\bmod2=1$. \end{proof}

From now on we consider only $m$, $n$ obeying the non-tiling conditions.
Recall the following definitions. Let $\varphi_{mn}$ be an eigenfunction
of the form \eqref{eq:eigenfunction_appendix} and \begin{align*}
\varphi_{mn}^{1}(x) & =\sin(mx)\sin(ny)\\
\varphi_{mn}^{2}(x) & =\sin(nx)\sin(my).\end{align*}
 Furthermore \begin{align*}
N_{m,n}^{1} & =\left\{ (x,y)\in\mathcal{D}\,\big|\, x\in\frac{\pi}{m}\mathbb{N}\vee y\in\frac{\pi}{n}\mathbb{N}\right\} ,\\
N_{m,n}^{2} & =\left\{ (x,y)\in\mathcal{D}\,\big|\, x\in\frac{\pi}{n}\mathbb{N}\vee y\in\frac{\pi}{m}\mathbb{N}\right\} \mbox{ and }\\
V_{m,n} & =\left\{ \frac{\pi}{m}(i,j)|0<j<i<m\right\} \cup\left\{ \frac{\pi}{n}(i,j)|0<j<i<n\right\} .\end{align*}
 By $\mathcal{N}(\varphi_{mn})$ we denote the nodal set of $\varphi_{mn}$.
Let $\mathcal{I}_{c}\subset\mathcal{D}\setminus(N_{m,n}^{1}\cup N_{m,n}^{2})$
be a rectangular shaped cell whose boundary is contained in $N_{m,n}^{1}\cup N_{m,n}^{2}$
and contains $c$ points from $V_{m,n}$, with $p_{0}$ being its
centre point. We also assume that $\forall(x,y)\in\mathcal{I}_{c}\,:\,\sign\varphi_{mn}^{1}(x,y)=\sign\varphi_{mn}^{2}(x,y)$.
\begin{lemma}\label{lemapp1} \global\long\def\theenumi{\roman{enumi}}
 \global\long\def\labelenumi{\theenumi}
 \hspace*{0.5cm}\\[-0.3cm] 
\begin{enumerate}
\item $\mathcal{N}(\varphi_{mn})\cap\mathcal{I}_{2}$ consists of a non-self-intersecting
line connecting the nodal corners of $\mathcal{I}_{2}$. 
\item $\pm\varphi_{mn}(p_{0})>0$ and $\mathcal{N}(\varphi_{mn})\cap\mathcal{I}_{4}$
consists of two separated lines each connecting adjacent nodal corners
along edges with $\mp\varphi_{mn}>0$. 
\end{enumerate}
\end{lemma} \begin{proof} Nodal sets on 2-dimensional manifolds
are submanifolds except for a closed set of lower dimension, where
nodal lines intersect. For an eigenfunction $\varphi_{mn}$ this singular
set is characterized by $\varphi_{mn}^{-1}(0)\cap(\nabla\varphi_{mn})^{-1}(0)$.
The boundary of a rectangle $\mathcal{I}_{2}$ with two points of
$V_{mn}$ intersects the nodal set only at those two points. By elementary
arguments using the monotonicity of the $\sin$ function, the existence
of nodal lines that do not intersect with the boundary of this rectangle
can be ruled out. The nodal set has to connect the nodal corners,
since nodal lines do not end. We present this argument in detail for
one specific case and leave the other cases to the reader. We consider
the situation of figure \ref{figure-appendix-1}. 

\begin{figure}
\begin{picture}(100,150) \put(-40,-130){\includegraphics[width=7cm]{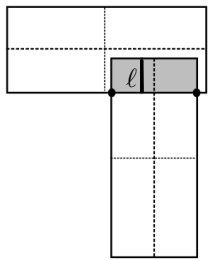}}
\put(47,69){$\ell$} \end{picture} \caption{Example of a superposition of the nodal pattern of $\varphi^{1}$
and $\varphi^{2}$}

\label{figure-appendix-1} %
\end{figure}
Let the rectangle $\mathcal{I}_{2}$ be in this case such that only
one symmetry axis of the two nodal domains of $\varphi_{mn}^{1}$
and $\varphi_{mn}^{2}$ enters $\mathcal{I}_{2}$. The symmetry axes
are the dotted lines and $\mathcal{I}_{2}$ is shaded. The vertices
on the lower corners belong to the nodal set and the boundary of $\mathcal{I}_{2}$
between those two points belongs to a nodal domain of $\varphi_{mn}$
with positive sign (assume this for now - for negative sign it would
be the same argument). Then the upper boundary of $\mathcal{I}_{2}$
belongs to a nodal domain of $\varphi_{mn}$ with a negative sign.
We study now the behaviour of $\varphi$ on a vertical line $\ell$
between the upper and lower boundary - like the one displayed in the
figure. Note first that the horizontal rectangle in the figure is
a nodal domain of $\varphi_{mn}^{2}$ with positive sign while the
vertical rectangle is a nodal domain of $\varphi_{mn}^{1}$ with positive
sign. On the lower end of $\ell$, $\varphi_{mn}^{2}$ starts equal
to zero and grows strictly monotonic on $\ell$ until it reaches the
boundary at a positive value. $\varphi_{mn}^{1}$ starts with a positive
value and falls strictly monotonic ending at zero. $\varphi_{mn}$
being the difference of $\varphi_{mn}^{1}$ and $\varphi_{mn}^{2}$
equals zero exactly once on $\ell$. This is true for any $\ell$.
The nodal set therefore intersects every $\ell$ exactly once and
therefore has no intersections nor further isolated nodal domains.\\
 In the case of rectangles with 4 points of $V_{mn}$ there is
a line of constant sign of $\varphi_{mn}$ running through the centre,
which cannot be intersected by a nodal line. It can be concluded as
above that the nodal corners are joined by nodal lines within the
two remaining components of this rectangle. \end{proof}

Let $\mathcal{T}\subset\mathcal{D}\setminus(N_{m,n}^{1}\cup N_{m,n}^{2})$
be a triangular shaped cell next to the boundary with $\sign\varphi_{mn}^{1}(x,y)=\sign\varphi_{mn}^{2}(x,y)$
in $\mathcal{T}$. Let $\mathcal{I}^{b}\subset\mathcal{D}\setminus(N_{m,n}^{1}\cup N_{m,n}^{2})$
be a rectangular shaped cell next to the boundary with $\sign\varphi_{mn}^{1}(x,y)=\sign\varphi_{mn}^{2}(x,y)$
in $\mathcal{I}^{b}$.

\begin{lemma}\label{lemapp2} \global\long\def\theenumi{\roman{enumi}}
 \global\long\def\labelenumi{\theenumi}
 \hspace*{0.5cm}\\[-0.3cm] 
\begin{enumerate}
\item A triangular cell $\mathcal{T}$ contains a nodal line iff $\overline{\mathcal{T}}$
contains a point of $V_{m,n}$. $\mathcal{N}(\varphi_{mn})\cap\mathcal{T}$
is a nodal line connecting this point to the boundary. 
\item If $\overline{\mathcal{I}^{b}}$ contains one point of $V_{m,n}$
then $\mathcal{N}(\varphi_{mn})\cap\mathcal{I}^{b}$ is a nodal line
connecting this point to the boundary. 
\item If $\overline{\mathcal{I}^{b}}$ contains two points of $V_{m,n}$
then $\mathcal{N}(\varphi_{mn})\cap\mathcal{I}^{b}$ is a nodal line
connecting those two points. 
\end{enumerate}
\end{lemma} \begin{proof}In order to understand the run of the nodal
set, the nodal pattern is continued beyond the hypotenuse by defining
the eigenfunction on the whole square to be the continuation of the
eigenfunction on the triangle. This rectangle can now be treated just
as in Lemma \ref{lemapp1} with two points of $V_{m,n}$ on the left
lower and right upper corner. The resulting nodal line coincides with
the hypotenuse. In case there is a point of $V_{m,n}$ on the right
lower corner, there is also one in the left upper corner by symmetry
and the case with four nodal corners from Lemma \ref{lemapp1} applies,
and shows the existence of a nodal line connecting the right lower
corner with the boundary. The other points are proven similarly to
the proof of lemma \ref{lemapp1} by monotonicity of $\varphi^{1}$
and $\varphi^{2}$. \end{proof}

\end{document}